\documentclass[11pt,aps,longbibliography,superscriptaddress]{revtex4-2}

\usepackage{epsfig,dsfont,amssymb,amsmath,amsthm,amsfonts,amsbsy,mathrsfs}
\usepackage{graphicx, color}
\usepackage{epstopdf}
\usepackage{mathdots}
\usepackage{float}
\usepackage{xcolor}
\usepackage{chemfig}

\usepackage{booktabs}

\usepackage{tikz}
\usetikzlibrary{quantikz2}

\usepackage{subcaption}
\usepackage[T1]{fontenc}
\usepackage{bm} 
\usepackage{esint}
\usepackage{mathrsfs}
\usepackage{booktabs}
\usepackage{multirow}
\usepackage{makecell}
\usepackage{verbatim}
\usepackage[
    colorlinks=true,
    linkcolor=blue,
    urlcolor=blue,
    citecolor=blue
]{hyperref}

\definecolor{revisionbrown}{RGB}{150,75,0}

\usepackage{appendix}
\usepackage{times}{\Large}

\newcommand{\ketbra}[2]{| #1 \rangle\!\langle #2 |}

\newtheorem{theorem}{Theorem}
\newtheorem{lemma}{Lemma}
\newtheorem{Proposition}{Proposition}

\newcommand{\Lie}{\mathrm{Lie}}
\newcommand{\ad}{\operatorname{ad}}
\newcommand{\im}{\operatorname{im}}

\makeatother

\begin{document}


\title{A Dynamical Lie-Algebraic Framework for Hamiltonian Engineering and Quantum Control}


\author{Yanying Liang}
\affiliation{College of Mathematics and Informatics, South China Agricultural University, Guangzhou, 510640, China}
\affiliation{Institute of High Performance Computing (IHPC), Agency for Science, Technology and Research (A*STAR), 1 Fusionopolis Way, \#16-16 Connexis, Singapore 138632, Singapore}

\author{Ruibin Xu}
\affiliation{College of Mathematics, South China University of Technology, Guangzhou, 510640, China}

\author{Mao-Sheng Li}
\affiliation{College of Mathematics, South China University of Technology, Guangzhou, 510640, China}

\author{Haozhen Situ}
\email{situhaozhen@gmail.com}
\affiliation{College of Mathematics and Informatics, South China Agricultural University, Guangzhou, 510640, China}

\author{Zhu-Jun Zheng}
\affiliation{College of Mathematics, South China University of Technology, Guangzhou, 510640, China}
\affiliation{Laboratory of Quantum Science and Engineering, Guangzhou, 510640, China}

\date{\today}

\begin{abstract}
Determining the unitary dynamics accessible from finite Hamiltonian resources
is a central problem in Hamiltonian engineering and quantum control.
Dynamical Lie algebras (DLAs) connect available control Hamiltonians with the
reachable dynamics, but their use as a design tool for modifying Hamiltonian
generator sets remains less developed. In this work, we develop a
finite-dimensional DLA framework for three generator-set operations:
composition, invariance, and reduction. For composition, we construct direct
sums of component DLAs using spectral projectors on an auxiliary register. For
invariance, we analyze when modifications of Pauli-string generating sets
preserve the generated Lie algebra, and introduce algebraic diagnostics for
added generators. For reduction, we consider compact reductive DLAs and show
how projection onto selected simple ideals gives reduced generating sets whose
Lie closures are the corresponding ideal sums. We illustrate these results
with finite-dimensional examples and numerical checks, including direct-sum
dimension addition, central-spin invariance diagnostics, and DLA-based ansatz
reduction for block-local Hamiltonians. The results show how DLA structure can
be used to diagnose controllability and guide Hamiltonian generator design
under explicit algebraic assumptions.

\end{abstract}

\maketitle


\section{Introduction}
Engineering the unitary dynamics generated by controllable Hamiltonians is a central problem in quantum computation, Hamiltonian simulation, and quantum control theory \cite{kallush2022controlling,jin2025universal,ramakrishna1996relation,schirmer2002identification,d2021introduction}. Given a finite set of accessible Hamiltonians, the central question is to determine which unitary transformations can be generated through their evolution and in what ways these dynamics may be optimized for special tasks. This question is fundamental to the design of resource-efficient quantum
algorithms in quantum computing
\cite{wu2026quantifying,ye2026coherence}, and also crucial for the
understanding of controllable quantum systems in quantum mechanics
\cite{ramakrishna1996relation,gago2023determining}.

A mathematical description of the achievable dynamics is provided by the
dynamical Lie algebra (DLA)
\cite{d2021introduction,d2025controllability}, denoted as
\(\mathfrak{g}_\mathcal{A}\). It is formed by the nested commutators and
real linear combinations of its generator set \(\mathcal{A}\).
The structure of \(\mathfrak{g}_\mathcal{A}\) determines the reachable set of
unitary evolutions and therefore governs the expressivity, controllability,
and symmetry properties of Hamiltonian-driven quantum dynamics
\cite{ragone2024lie,fontana2024characterizing,mcclean2018barren,larocca2025barren,wiebe2014hamiltonian,DavidEdwardBruschi2025}.
As a result, DLAs have the potential to form a unifying framework for
analyzing a wide range of quantum processes.

Recent research has made substantial progress in characterizing DLAs associated with specific Hamiltonian structures, including Pauli-string generators \cite{aguilar2024full}, labeled directed graphs \cite{DavidEdwardBruschi2026}, and specialized variational circuits such as the Quantum Approximate Optimization Algorithm \cite{allcock2026dynamical,kokcu2024classification,kazi2025analyzing,rabinovich2025role}.
Despite these advances, there is still a limited understanding of how algebraic modifications to the generator set $\mathcal{A}$ fundamentally impact the resulting DLA $\mathfrak{g}_\mathcal{A}$. Understanding this relationship is essential for systematically shaping DLA structures to achieve desirable controllability, symmetry, and expressive properties.

Recent research has begun to touch on this issue. Allcock et al. \cite{allcock2025generating} proposed qubit- and parameter-efficient constructions of DLAs, while a related study \cite{zimboras2015symmetry} highlighted the importance of symmetry and algebraic structure for DLA simulability. More recently, Smith et al. \cite{smith2025optimally} demonstrated that for generating sets $\mathcal{A}$ consisting solely of Pauli strings, the minimal number of operators required to generate $\mathfrak{su}(2^N)$ is \(2N + 1\). These findings suggest the potential of strategically modifying generator sets to shape DLA structures for targeted controllability. Despite these advances, three related questions need to be considered, as summarized in Fig.~\ref{Figure-3question} and discussed in Ref.~\cite{allcock2025generating}:
\begin{itemize}
    \item \textbf{Question 1}: Given two generator sets $\mathcal{A}$ and $\mathcal{B}$ generating $\mathfrak{g}_\mathcal{A}$ and $\mathfrak{g}_\mathcal{B}$, what generator set $\mathcal{C}$ yields $\mathfrak{g}_\mathcal{C} \cong \mathfrak{g}_\mathcal{A} \oplus \mathfrak{g}_\mathcal{B}$?
    \item \textbf{Question 2}: Under what modifications $\mathcal{A} \to \mathcal{A}'$ does the DLA remain unchanged, i.e., $\mathfrak{g}_\mathcal{A} = \mathfrak{g}_{\mathcal{A}'}$?
    \item \textbf{Question 3}: For a specified Lie structure \(s \subseteq \mathfrak{g}_\mathcal{A}\), under what conditions can \(\mathcal{A}\) be modified to \(\mathcal{A}'\) such that \(\mathfrak{g}_{\mathcal{A}'} \cong s\)?
\end{itemize}

\begin{figure}[!t]
\centering
\includegraphics[width=13cm]{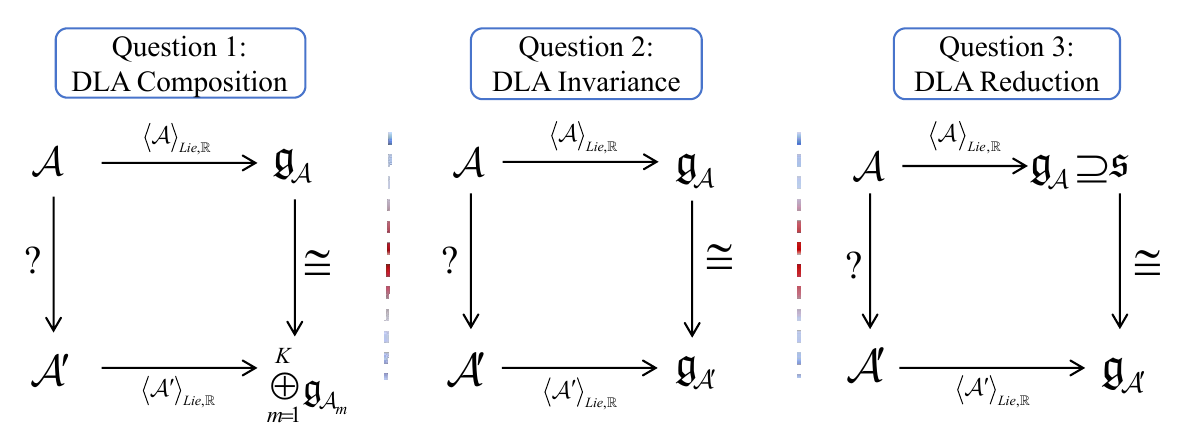}
\caption{ Illustration of three questions on modifying generating set $\mathcal{A}$ of a DLA $\mathfrak{g}_\mathcal{A}$: DLA composition, DLA invariance, and DLA reduction. }
\label{Figure-3question}
\end{figure} 

Question 1 addresses how a larger DLA can be constructed as a direct sum of
smaller, task-specific DLAs. 
A naive approach is to place each generator set \(\mathcal{A}_m\)
\((m=1,2,\dots,K)\) on an independent \(n\)-qubit register, which would require
\(nK\) qubits to realize
\(\mathfrak{g}_{\mathcal{A}'} \cong
\bigoplus_{m=1}^K \mathfrak{g}_{\mathcal{A}_m}\).
This may become impractical for resource-constrained quantum computing.
Here, resource constrained refers to the limited size of the available
system register. In this sense, reducing the number of required system qubits is an important consideration for implementing multiple dynamical components on near-term quantum devices.
Question 1 thus seeks efficient methods to realize the direct sum of
different DLAs, which can be applied in parallel representations of
Hamiltonian-driven dynamics, whether time dependent or time independent, on small noisy
intermediate-scale quantum computers
\cite{barratt2021parallel,diaz2025parallel}.

Question 2 concerns the context of quantum circuit optimization \cite{nam2018automated,iten2022exact}. Specifically, it investigates transformations $\mathcal{A} \to \mathcal{A}'$ that preserve the DLA unchanged. Let $|\mathcal{A}|$ denote the cardinality of the generator set $\mathcal{A}$. If \(|\mathcal{A}| = |\mathcal{A}'|\), the circuit corresponding to \(\mathcal{A}'\) realizes the same algebraic structure with identical cardinality, but possibly enhances qubit efficiency or reduces circuit complexity through reorganization of the generator set $\mathcal{A}$, especially when $\mathcal{A}$ consists of Pauli strings \cite{viswanathan2025}.
If \(|\mathcal{A}| <|\mathcal{A}'|\), the transformation can help simulate one Hamiltonian interaction set by another, 
instead of the brute-force calculation of nested commutators \cite{zimboras2015symmetry}.

Question 3 aims to systematically restrict a quantum system's accessible
dynamics to a lower-dimensional Lie structure inside
\(\mathfrak{g}_\mathcal{A}\). 
Addressing Question 3 can enable the precise engineering of dynamical
evolutions in decoherence-free subspaces, where error generators are linked to the generators of Lie algebras
\cite{PhysRevLett.81.2594,xu2012nonadiabatic}. Moreover, inspired by
Ref.~\cite{bosse2025efficient}, Question 3 can also motivate
resource-efficient approximations of large quantum systems by representing
the relevant dynamics within low-dimensional, well-structured algebraic
components. Thus, when the full system suffers from a barren plateau, such
reduced algebraic structures may provide a pathway to compare and design more
trainable circuit ansatz families under explicit Lie-algebraic assumptions
\cite{ragone2024lie,fontana2024characterizing}.

Motivated by the significance and potential applications of these three
questions, we answer Questions 1 and 3, and offer a specialized analysis of
\(\mathfrak{su}(2^N)\) and the DLAs with
\(|\mathcal{A}|<|\mathcal{A}'|\) in Question 2. We do not attempt to solve
these problems in full generality; rather, we provide constructive results in
physically relevant finite-dimensional settings. In Question 1,
we use the spectral decomposition of a Hermitian operator to link different
DLAs while ensuring linear independence among the \(K\) direct-sum components,
leading to a register-efficient direct-sum construction. For Question 2, when
\(|\mathcal{A}|=|\mathcal{A}'|\), we present a new Pauli-string generating set
for \(\mathfrak{su}(2^N)\), which is relevant to circuit compilation; when
\(|\mathcal{A}|<|\mathcal{A}'|\), we give two quantitative
algebraic diagnostics that address a limitation in
Ref.~\cite{zimboras2015symmetry}. In Question 3, we work in
a compact reductive setting and construct reduced generating sets by
projecting onto selected simple ideal components, with finite-generator
examples illustrating when filtered constructions succeed or fail.

The structure of this manuscript is outlined as follows. Section~2 provides the concepts from Lie theory and related studies. Section~3 presents theoretical results addressing these three questions. Section 4 gives the related applications and numerical studies for our findings.
In Section~5, we summarize these results for three questions and discuss potential future work.

\section{Preliminary}
\label{sec:preliminary}

\subsection{Background}

\textbf{The controllability of a quantum system.}
Consider an \(N\)-qubit quantum system with Hilbert space
\(\mathcal H=(\mathbb C^2)^{\otimes N}\).
Let \(\rho=\ketbra{\psi}{\psi}\) be the
pure-state density operator corresponding to
\(|\psi\rangle\in\mathcal H\). The evolution of \(|\psi\rangle\) is
determined by the Schr\"odinger equation
\cite{diaz2025parallel,smith2025optimally}
\begin{equation}\label{sch-equ}
i\hbar\frac{d}{dt}|\psi\rangle=H|\psi\rangle,
\end{equation}
where we set Planck's constant \(\hbar=1\). The solution of
Eq.~\eqref{sch-equ} is well known to be
\begin{equation}\label{evolution-1H}
|\psi_1\rangle=U(t)|\psi_0\rangle,\quad U(t)=e^{-iHt},
\end{equation}
for some initial state \(|\psi_0\rangle\).
Consider a set of controllable Hamiltonians \(\{H_i\}_{i=1}^L\), which
correspond to different basic gates in quantum computing. Since each
\(H_i\) is time independent, the evolution under a single \(H_i\) is a
unitary transformation. A piecewise-constant control sequence gives
\begin{equation}\label{evolution-many-H}
|\psi_1\rangle=e^{-iH_Lt_L}\cdots e^{-iH_1t_1}|\psi_0\rangle,
\end{equation}
for real parameters \(t_1,\cdots,t_L\).
More general time-dependent controls may be treated through
piecewise-constant approximations. Here Eq.~\eqref{evolution-many-H} is
used only to identify the available Hamiltonian directions considered
below.

\begin{figure}[htbp]
\centering
\includegraphics[width=13cm]{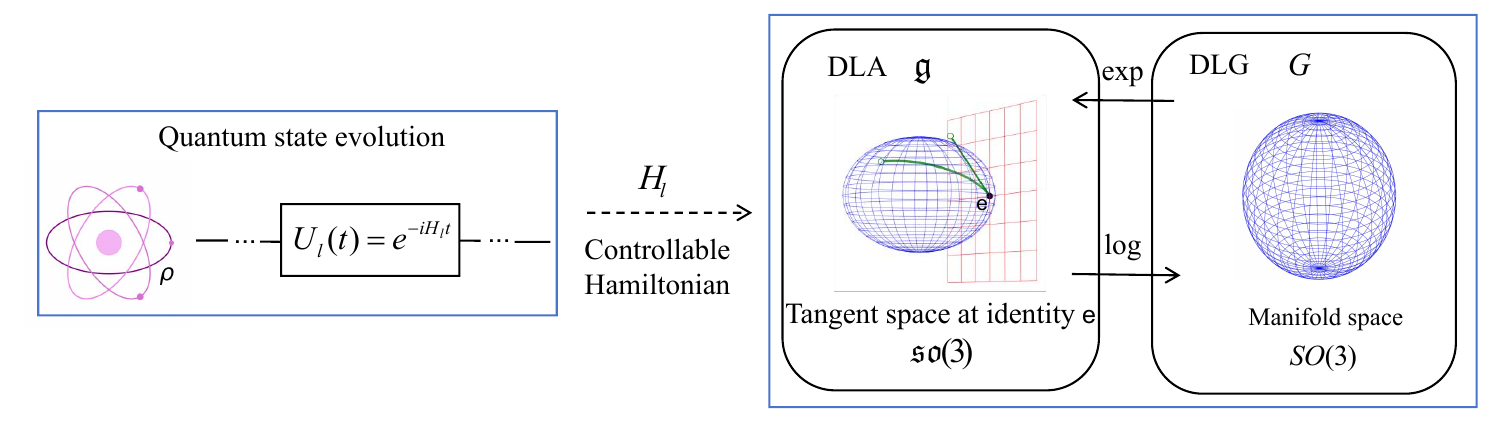}
\caption{Lie-theoretic background. A density operator evolves as
\(\rho(t)=U(t)\rho(0)U^\dagger(t)\), with accessible unitaries
\(U(t)\in G\). The DLA \(\mathfrak g=\operatorname{Lie}(G)\) is the
tangent space of \(G\) at the identity, and the exponential map connects
\(\mathfrak g\) to \(G\).}
\label{Figure-lietheory}
\end{figure} 

\textbf{Controllability and DLAs.}
Collecting the associated anti-Hermitian generators described above, we
define
\begin{equation}\label{generatingset}
\mathcal A=\{A_1,A_2,\cdots,A_L\}.
\end{equation}
Throughout the finite-dimensional parts of this work, unless otherwise
stated, all DLAs are real Lie subalgebras of \(\mathfrak{su}(2^N)\) for an \(N\)-qubit system. 
The \textit{dynamical Lie algebra} (DLA) generated by \(\mathcal A\) is
\cite{d2021introduction}
\begin{equation}
\mathfrak g_{\mathcal A}
=
\langle\mathcal A\rangle_{\mathrm{Lie},\mathbb R},
\label{orginialdla}
\end{equation}
which is the real vector space spanned by all possible nested commutators
of elements of \(\mathcal A\). The set \(\mathcal A\) in
Eq.~\eqref{generatingset} is called a \textit{dynamical generating set}
for \(\mathfrak g_{\mathcal A}\).
Signs and nonzero scalar factors generated by commutators do not affect
the DLA, since the Lie closure is taken over the real span. Throughout this paper, \(=\) denotes equality as represented subalgebras
of a fixed matrix algebra, whereas \(\cong\) denotes abstract Lie-algebra
isomorphism.

\textbf{DLA and reachable Lie group.}
Following the standard formulation of finite-dimensional quantum control
systems, the DLA \(\mathfrak g_{\mathcal A}\) determines the connected Lie
subgroup \(G_{\mathcal A}\subseteq SU(2^N)\) on which the controlled
dynamics evolves \cite{schirmer2002identification}. We write
\(\operatorname{Lie}(G_{\mathcal A})=\mathfrak g_{\mathcal A}.\)
We refer to \(G_{\mathcal A}\) as the associated dynamical Lie group
(DLG). The actually reachable subset may depend on the
allowed control amplitudes and time schedules.
As a smooth manifold,
\(\dim(G_{\mathcal A})=\dim(\mathfrak g_{\mathcal A})\).

The significance of DLAs lies in their ability to characterize the
ultimate controllability of quantum state evolution. The accessible unitaries generated by the available controls satisfy
\(U(t)\in G_{\mathcal A}\) for suitable control parameters \cite{ragone2024lie}.
Fig.~\ref{Figure-lietheory} provides a schematic illustration of the
Lie-theoretic background in this paper.

\subsection{Related works}

Modifying the generating sets of DLAs can provide a powerful way to construct new DLAs efficiently. Several recent studies have explored this direction from the following perspectives.

\begin{lemma} [\cite{allcock2025generating}]\label{lemma1}
 Consider an initial generating set $\mathcal{A}=\{A_1,\cdots,A_L\}$. Let $\chi$ be a Hermitian operator having $K$ distinct eigenvalues. Define $\mathcal{A}' = \{A_i \otimes \chi^j | i \in \{1,2,\cdots,L\}, j \in \{0,1,\cdots, K-1\} \}$. Then we have $\mathfrak{g}_{\mathcal{A}'} \cong \bigoplus_{j=1}^{K} \mathfrak{g}_{\mathcal{A}}$.   
\end{lemma}

Lemma \ref{lemma1} shows that coupling the generating set $\mathcal{A}$ with a Hermitian operator $\chi$ having $K$ distinct eigenvalues yields a new generating set $\mathcal{A}'$, whose DLA is a direct sum of $K$ identical copies. This approach is qubit-efficient, utilizing only \( \lceil \log K \rceil \) additional qubits.
However, regarding Question 1, this method is restricted in a direct sum of identical DLAs.  Consequently, it is insufficient for flexible quantum circuit design that may require different components.

\begin{lemma}[\cite{smith2025optimally}]\label{lemma3}
 Fix \(2 \leq k \leq N \) and consider \(\mathcal{A} \subset i P_{k}^{*}\) and \(\mathcal{B} \subset i P_{N-k}^{*}\) such that \(\langle\mathcal A\rangle_{Lie,\mathbb{R}} = \mathfrak{su}(2^{k})\) and \(\mathcal{B}\) is product universal for \(i P_{N-k}^{*}\). Let \(\mathbb{B} := \{(U_{B}, B)\}_{B \in \mathcal{B}}\) be a set of pairs where each \(U_{B} \in P_{k}^{*}\). Defining \(\mathcal{A}' = \{A \otimes I^{\otimes N-k} \mid A \in \mathcal{A}\}\) and \(\mathcal{B}' = \{U_{B} \otimes B \mid (U_{B}, B) \in \mathbb{B}\}\), we have that \(\langle\mathcal{A}' \cup \mathcal{B}'\rangle_{Lie,\mathbb{R}} = \mathfrak{su}(2^{N})\). $P_{k}^{*}$ denotes the set of Pauli strings for $k$ qubits, excluding the identity operator $I^{\otimes k}$.
\end{lemma}

Lemma \ref{lemma3} implies that the minimum number of elements in a Pauli-string generating set for the full $\mathfrak{su}(2^N)$ is $2N+1$. If the number of elements in the generating set attains this minimal size, one can
decompose the original generating set into a small-sized universal subset $\mathcal{A}$ and a product-universal subset for the remaining part $\mathcal{B}$, and then form new generating sets  \(\mathcal{A}'\)  and $\mathcal{B}'$ through tensor product extension with the identity matrix and non-identity Pauli strings. After modification, \(|\mathcal{A}'| = |\mathcal{A}| + |\mathcal{B}|\), where the size remains \(2N+1\) without changing. Via decomposition and extension, Lemma \ref{lemma3} ensures no increase in $|\mathcal{A}'|$ and provides a way to modify the generating set for $\mathfrak{su}(2^N)$, which is related to Question 2. Based on this lemma, we will present a novel construction for the generating set of $\mathfrak{su}(2^N)$ with \(2N+1\) generators.

\begin{lemma}[\cite{zimboras2015symmetry}]\label{lemma2}
 Let $\mathcal{P}=\{iH_1,\cdots,iH_p\}$ and $\mathcal{Q}=\{iH_{p+1},\cdots,iH_q\}$ be two sets of anti-Hermitian operators. $\mathcal{P}$ can simulate $\mathcal{Q}$, 
meaning that \(\langle \mathcal{P} \rangle_{Lie,\mathbb{R}} = \langle \mathcal{P} \cup \mathcal{Q} \rangle_{Lie,\mathbb{R}}\) if and only if 
(1) \( \dim[\mathcal{P}^{(2)}] = \dim[(\mathcal{P} \cup \mathcal{Q})^{(2)}] \) and 
(2) the central projections of $\mathcal{P}$ and $\mathcal{P} \cup \mathcal{Q}$ onto $\mathcal{C}$ are of the same rank, here $(\mathcal{P} )'$ represents the commutant of $(\mathcal{P} )$, 
$\mathcal{C}$ is the center of the commutant $(\mathcal{P} \cup \mathcal{Q})'$, and  \(P^{(2)} = (P^{\otimes 2})' = \{ S \in \mathbb{C}^{d^2 \times d^2} \mid [S, M \otimes I_d + I_d \otimes M] = 0, \forall  M \in P \}\),$d=2^N$ in $N$-qubit system.   
\end{lemma}

Lemma~\ref{lemma2} gives an algebraic criterion for deciding when two anti-Hermitian generating sets
\(\mathcal A=\mathcal P\) and \(\mathcal A'=\mathcal P\cup\mathcal Q\) generate the same DLA. 
It therefore provides a useful reference criterion for exact DLA invariance in Question~2, in the case where the generating set is enlarged. 
In the following discussion, we use this exact criterion as a baseline and then introduce diagnostic quantities for cases where exact invariance is not satisfied in Lemma~\ref{lemma2}.

Regarding Question 3, to the best of our knowledge, no constructions of generating sets have been proposed for DLA reduction. This paper aims to address this gap in the compact finite-dimensional setting.

\section{Main results}
 By investigating the relationship between generator selection and the resulting algebraic structure, we provide detailed answers to the three questions outlined in Fig.~\ref{Figure-3question}. 

\subsection{DLA composition via direct sum construction}

We begin with Question 1 investigating the composition of DLAs. Starting with the binary case of merging $\mathfrak{g}_\mathcal{A}$ and $\mathfrak{g}_\mathcal{B}$, we extend this to the construction of a DLA isomorphic to the direct sum of 
$K$ distinct DLAs $\mathfrak{g}_1, \mathfrak{g}_2, \ldots, \mathfrak{g}_K$.
Let \(\mathcal H_S=(\mathbb C^2)^{\otimes N}\) denote the Hilbert space of an \(N\)-qubit system, and let \(\mathcal H_A\) denote an auxiliary Hilbert space.

\begin{theorem}\label{composition1}
 For each
\(m=1,2,\ldots,K\), let
\(\mathcal A_m
    =
    \{A_{m,1},A_{m,2},\ldots,A_{m,L_m}\}
    \subset \mathfrak{su}(\mathcal H_S)\)
be a dynamical generating set acting on the same system Hilbert
space \(\mathcal H_S\), and let
\(\mathfrak g_{\mathcal A_m}
    =
    \langle \mathcal A_m\rangle_{\mathrm{Lie},\mathbb R}.\)
Suppose that
\(\chi\) is a Hermitian operator on \(\mathcal H_A\) with
\(K\) distinct eigenvalues \(\lambda_1,\lambda_2,\ldots,\lambda_K\),
and denote by \(\Pi_1,\Pi_2,\ldots,\Pi_K\) the corresponding orthogonal
spectral projectors. These projectors satisfy
\(\Pi_m\Pi_n=\delta_{mn}\Pi_m,
    \sum_{m=1}^K \Pi_m=I,\)
where \(I\) denotes the identity operator on \(\mathcal H_A\).
Define
\(\mathcal A'
    =
    \left\{
        A_{m,l}\otimes\Pi_m:
        m=1,\ldots,K,\;
        l=1,\ldots,L_m
    \right\}
    \subset \mathfrak{su}(\mathcal H_S\otimes\mathcal H_A).\)
Then
\[
    \mathfrak g_{\mathcal A'}
    \cong
    \bigoplus_{m=1}^K \mathfrak g_{\mathcal A_m}.
\]
\end{theorem}

Theorem~\ref{composition1} shows that a direct-sum DLA can be constructed
from component DLAs represented on a common system Hilbert space. The proof
is given in Appendix~\ref{app:proof-composition}. 
The construction should be interpreted at the level of represented
generator sets. All component generators \(A_{m,l}\) act on the same system
Hilbert space \(\mathcal H_S\), while the projectors \(\Pi_m\) act on the
auxiliary Hilbert space \(\mathcal H_A\). If the component DLAs are
originally represented on Hilbert spaces of different dimensions, fixed
faithful embeddings into a common matrix algebra should be chosen before
applying Theorem~\ref{composition1}.

Comparing Lemma~\ref{lemma1} with Theorem~\ref{composition1}, the two
constructions use, respectively,
\begin{align}\label{A'inlemma1}
 \mathcal{A}' =
 \{A_i \otimes \chi^j \mid
 i \in \{1,2,\cdots,L\},\,
 j \in \{0,1,\cdots,K-1\} \},
\end{align}
which employs powers of a single Hermitian operator \(\chi\), and
\begin{align}\label{A'inthm1}
   \mathcal{A}' =
   \left\{
        A_{m,l}\otimes\Pi_m:
        m=1,\ldots,K,\;
        l=1,\ldots,L_m
   \right\},
\end{align}
which instead uses the orthogonal spectral projectors \(\Pi_m\) associated
with the eigenspaces of \(\chi\) as auxiliary labels. 
In Eq.~\eqref{A'inthm1}, the index \(m\) labels the component DLA, while
\(l\) labels the generators inside the \(m\)-th component set
\(\mathcal A_m\).

For a Hermitian operator \(\chi\) with \(K\) distinct eigenvalues, the set
of powers \(\{\chi^0,\chi^1,\ldots,\chi^{K-1}\}\) is linearly independent
and spans the same space as the spectral projectors
\(\{\Pi_1,\Pi_2,\ldots,\Pi_K\}\). Indeed, since
\[
    \chi^j=\sum_{i=1}^K \lambda_i^j\Pi_i,
    \quad j\geq 0,
\]
the Vandermonde matrix associated with the distinct eigenvalues
\(\lambda_1,\ldots,\lambda_K\) is invertible. Hence
\[
    \operatorname{span}\{\chi^0,\ldots,\chi^{K-1}\}
    =
    \operatorname{span}\{\Pi_1,\ldots,\Pi_K\}.
\]
This explains why the projector-based construction uses the auxiliary
spectral labels directly.

The eigenvalues \(\lambda_1,\lambda_2,\ldots,\lambda_K\) are required to
be distinct so that they define \(K\) different spectral projectors.
However, the corresponding eigenspaces need not be one-dimensional. For
example, if
\(\Pi_m=|0\rangle\langle 0|+|1\rangle\langle 1|,\)
then \(A_{m,l}\otimes\Pi_m\) applies the same system generator
\(A_{m,l}\) on both auxiliary states \(|0\rangle\) and \(|1\rangle\).
Therefore, a degenerate
eigenspace gives one copy of \(\mathfrak g_{\mathcal A_m}\) with
multiplicity, not several independently controllable copies. In this
sense, the construction is parallel across different projectors
\(\Pi_m\), but not within the degeneracy of a single projector.
Equivalently, if \(\mathcal H_m\) denotes the eigenspace associated with
\(\Pi_m\), then \(\mathcal H_S\otimes\mathcal H_m\) is the corresponding
invariant branch subspace. On this subspace,
\(A_{m,l}\otimes\Pi_m\) acts as \(A_{m,l}\) on the system register and
identically on the auxiliary degeneracy space. Thus no Cartan-type
decomposition is assumed in this construction; only the spectral
decomposition of the auxiliary operator \(\chi\) is used.

When comparing with a naive allocation,
if \(K\) component systems are implemented on \(K\) independent
\(n\)-qubit system registers, the system-register allocation uses \(nK\)
qubits. By contrast, when all component DLAs are represented on a common
\(2^n\)-dimensional system space, the block-projector construction uses
one \(n\)-qubit system register together with an auxiliary label register.
This comparison concerns the qubit count for Hilbert-space allocation;
it is not a claim about Hamiltonian synthesis cost, control-term
count, locality, or platform-specific implementation.

As for the generator count, the displayed construction uses
\(|\mathcal A'|
    =
    \sum_{m=1}^K L_m\)
controlled generators.
 A smaller generating set may exist if some component algebra
admits a smaller generating set or if additional representation-dependent
relations are available. 

Finally, the direct-sum structure follows from
\(\Pi_m\Pi_n=\delta_{mn}\Pi_m\). Indeed, cross-branch commutators vanish
for \(m\neq n\), while for a fixed branch,
\[
    [A_{m,l}\otimes\Pi_m,A_{m,k}\otimes\Pi_m]
    =
    [A_{m,l},A_{m,k}]\otimes\Pi_m .
\]
Therefore, different auxiliary branches commute and each branch closes
under the Lie brackets of its own system generators. This yields the
direct-sum decomposition of the generated DLA.

\subsection{DLA invariance with generator replacement}

Let us now turn to Question 2: \textit{Under what conditions does a
modification \(\mathcal{A}\to\mathcal{A}'\) leave the DLA unchanged?}
Let \(\mathfrak{g}_{\mathcal A}\subseteq\mathfrak{su}(2^N)\) be the DLA
generated by a set of anti-Hermitian matrices
\(\mathcal{A}=\{A_1,\dots,A_L\}\) with cardinality \(L\).
Let \(\mathcal{A}'=\{A'_1,\dots,A'_{L'}\}\) be another finite set of
anti-Hermitian matrices with cardinality \(L'\). To systematically analyze
DLA invariance, we distinguish the cases based on the relationship between
\(L\) and \(L'\).

Following the convention stated in Sec.~\ref{sec:preliminary}, this
subsection works within the traceless anti-Hermitian setting
\(\mathfrak{su}(2^N)\).
In this subsection, we call a DLA modification invariant  when
\[
    \langle\mathcal A\rangle_{\mathrm{Lie},\mathbb R}
    =
    \langle\mathcal A'\rangle_{\mathrm{Lie},\mathbb R}.
\]
The exact criteria used below are those stated in Lemma~\ref{lemma3} and
Lemma~\ref{lemma2}.

\subsubsection{Generating set of Pauli strings with cardinality unchanged}

When \(L=L'\), a special case we choose is the full DLA
\(\mathfrak{su}(2^N)\). Lemma~\ref{lemma3} tells us that we can modify the
generating set \(\mathcal{A}\to\mathcal{A}'\) while preserving
\(\mathfrak{su}(2^N)\) without increasing the cardinality. Let
\(\mathcal{A}\) be a minimal universal generating set for
\(\mathfrak{su}(2^N)\) with \(L=2N+1\), and
\(\mathcal{A}\subset i\mathcal P_N^*\). Here, \(\mathcal P_N^*\) denotes
all Pauli strings on \(N\) qubits excluding the \(N\)-fold tensor product
of identity matrices \(I^{\otimes N}\). We can decompose \(\mathcal{A}\)
into \(\mathcal{A}_k\subset i\mathcal P_k^*\) with \(k\) qubits and
\(\mathcal{B}\subset i\mathcal P_{N-k}^*\) with \(N-k\) qubits for
\(2\leq k<N\). Define the modified set
\[\mathcal{A}'=\mathcal{A}_k'\cup\mathcal{B}',\] where
\(\mathcal{A}_k'
    =
    \{A\otimes I^{\otimes (N-k)}\mid A\in\mathcal{A}_k\},\)
extending \(\mathcal{A}_k\) to \(N\) qubits via identity tensor product,
and
\(\mathcal{B}'
    =
    \{U_B\otimes B\mid B\in\mathcal{B}\},\)
pairing each \(B\in\mathcal{B}\) with a Pauli string
\(U_B\in\mathcal P_k^*\). Then using Lemma~\ref{lemma3}, one gets
\[
    \langle\mathcal{A}\rangle_{\mathrm{Lie},\mathbb R}
    =
    \langle\mathcal{A}'\rangle_{\mathrm{Lie},\mathbb R}
    =
    \mathfrak{su}(2^N)
\]
with \(L=L'=2N+1\). In the following, we present a novel modification of
the generating set for \(\mathfrak{su}(2^N)\). 
For an \(N\)-qubit system, \(\sigma_j^\alpha\) denotes the \(N\)-qubit
Pauli operator acting as \(\sigma^\alpha\) on qubit \(j\) and as the
identity on all other qubits. For example, in a two-qubit system,
\(\sigma_1^x=\sigma^x\otimes I\) and
\(\sigma_2^x=I\otimes\sigma^x\). Products such as
\(\sigma_i^x\sigma_{i+1}^y\) are shorthand for tensor products with
identities on all remaining sites. Thus all operators appearing in
\(\mathcal A_2\) below are \(4\times4\) anti-Hermitian matrices.
 The first case is \(N=2\),
whose proof is in Appendix~\ref{proof-of-invariance-prop}.

\begin{Proposition}\label{prop:su2-case}
Let
\(\mathcal{A}_2
    =
    \{i\sigma^x_1,\ i\sigma^y_1,\ i\sigma^x_2,\ i\sigma^y_2,\ 
    i\sigma^y_1\sigma^y_2\},\)
then we have
\[
    \langle\mathcal{A}_2\rangle_{\mathrm{Lie},\mathbb R}
    =
    \mathfrak{su}(4).
\]
\end{Proposition}

Based on Proposition~\ref{prop:su2-case}, define
\begin{align}\label{b1}
\mathcal{B}_I'
=
\left\{
    i\sigma^x_2
    \left(\prod_{2<j<i}\sigma^z_j\right)
    \sigma^y_i,\;
    i\sigma^y_2
    \left(\prod_{2<j<i}\sigma^z_j\right)
    \sigma^x_i
    \;\middle|\;
    i=3,\dots,N
\right\}.
\end{align}
One can check that \(\mathcal{A}_2\) and \(\mathcal{B}_I'\) satisfy the
conditions in Lemma~\ref{lemma3}. Thus, set
\(\mathcal{A}=\mathcal{A}_2'\cup\mathcal{B}_I'\), with
\begin{align}\label{A2}
\mathcal{A}_2'
=
\{A\otimes I^{\otimes (N-2)}\mid A\in\mathcal{A}_2\}.
\end{align}
Then we have
\[
    \langle\mathcal{A}\rangle_{\mathrm{Lie},\mathbb R}
    =
    \mathfrak{su}(2^N).
\]
Next we show a modification of \(\mathcal{A}\) to \(\mathcal{A}'\) with
\(L=L'=2N+1\).

\begin{theorem}
\label{su2N-unchanged}
For the full DLA \(\mathfrak{su}(2^N)\) and \(\mathcal{A}_2'\) in
Eq.~\eqref{A2}, define the modified set
\(\mathcal{A}'=\mathcal{A}_2'\cup\mathcal{B}_{II}'\), where
\begin{align}\label{b2}
\mathcal{B}_{II}'
=
\left\{
    i\sigma^x_i\sigma^y_{i+1},\;
    i\sigma^y_i\sigma^x_{i+1}
    \;\middle|\;
    i=2,\dots,N-1
\right\},
\end{align}
with \(\sigma^x_i\sigma^y_{i+1}\) and
\(\sigma^y_i\sigma^x_{i+1}\) nearest-neighbor interactions. Then for \(L=L'=2N+1\), we have
\[
    \langle\mathcal{A}'\rangle_{\mathrm{Lie},\mathbb R}
    =
    \mathfrak{su}(2^N).
\]
\end{theorem}

Theorem~\ref{su2N-unchanged} provides a way to construct the generating
set for the full \(\mathfrak{su}(2^N)\) using \(\mathcal{B}_{II}'\) in
Eq.~\eqref{b2} instead of \(\mathcal{B}_I'\) in Eq.~\eqref{b1}. The
formal proof can be found in Appendix~\ref{proof-of-invariance-thm}. 
In quantum circuit optimization, 
nearest-neighbor interactions in
\(\mathcal{B}_{II}'\) are motivated by the connectivity constraints of
many quantum hardware platforms
\cite{linke2017experimental}. 
Since many platforms only permit direct operations
between adjacent qubits, circuits with non-adjacent interactions may
require SWAP insertions or additional gate decompositions, which can
increase circuit depth and error rates
\cite{zulehner2018efficient}. 
The actual compilation cost depends on the
hardware connectivity graph, native gate set, and compiler strategy.

\subsubsection{Generating set of Pauli strings with cardinality increasing}

Unlike the strict invariance in the 
\(L=L'\) case, for \(L<L'\), we propose two quantitative
diagnostic indicators to evaluate
the structural change induced by generator addition or replacement,
thereby addressing a practical limitation of Condition~(2)
in Lemma~\ref{lemma2}. Specifically, Lemma~\ref{lemma2} provides two
necessary and sufficient conditions for preserving
\begin{align}\label{lemma3-eq}
    \langle \mathcal{P} \rangle_{\mathrm{Lie},\mathbb{R}}
=
\langle \mathcal{P}\cup \mathcal{Q} \rangle_{\mathrm{Lie},\mathbb{R}}
\end{align}
with arbitrary finite generating sets
\(\mathcal{P}=\{iH_1,\cdots,iH_p\}\) and
\(\mathcal{Q}=\{iH_{p+1},\cdots,iH_q\}\) in the exploration of quantum
simulation.

Following the convention stated in Sec.~\ref{sec:preliminary}, the
generators below are regarded as elements of \(\mathfrak{su}(2^N)\).
We emphasize that the exact conclusion
in Eq.(\ref{lemma3-eq})
requires both conditions in Lemma~\ref{lemma2}: the condition involving
\(\mathcal{P}^{(2)}\) and the central-projection rank condition. The
quantities introduced below are therefore used as diagnostic indicators,
not as replacements for the full exact criterion.

For the exact criterion in Lemma~\ref{lemma2}, let
\(\mathcal{C}=Z((\mathcal{P}\cup\mathcal{Q})')\) be the center of the
commutant of \(\mathcal{P}\cup\mathcal{Q}\). Denote
\(|\mathcal{P}|=p\) and \(|\mathcal{P}\cup\mathcal{Q}|=q\).
Condition~(2) in Lemma~\ref{lemma2} states that the central projections of
\(\mathcal{P}\) and \(\mathcal{P}\cup\mathcal{Q}\) onto \(\mathcal{C}\)
are of the same rank. Assume that \(\mathcal{C}\) is a \(k\)-dimensional
linear space with an orthonormal basis
\(\{\rho_1,\rho_2,\cdots,\rho_k\}\)
with respect to the Hilbert--Schmidt inner product
\(\langle X,Y\rangle_{\mathrm{HS}}=\operatorname{Tr}(X^\dagger Y)\).
For any operator \(iH_\beta\in\mathcal{P}\cup\mathcal{Q}\), its central
projection can be expressed as
\(\Pi_{\mathcal C}(iH_\beta)
=
a_1\rho_1+a_2\rho_2+\cdots+a_k\rho_k,\)
where
\(a_\alpha=\operatorname{Tr}(\rho_\alpha^\dagger iH_\beta),
1\leq\alpha\leq k .\)
These coefficients form the projection matrix \(\tilde T\) for
\(iH_\beta\in\mathcal{P}\), \(1\leq\beta\leq p\), and the projection
matrix \(T\) for \(iH_\beta\in\mathcal{P}\cup\mathcal{Q}\),
\(1\leq\beta\leq q\).
Here, \(\tilde T\) and \(T\) are central projection coefficient
matrices: each column is the coefficient vector of the central projection
of one generator in the chosen Hilbert--Schmidt orthonormal basis.
The linear spaces formed by these projection results are denoted as
\(\operatorname{span}(\Pi_{\mathcal C}(\mathcal{P}))\) and
\(\operatorname{span}(\Pi_{\mathcal C}(\mathcal{P}\cup\mathcal{Q}))\), where \(\Pi_{\mathcal C}\) denotes the orthogonal projection onto
\(\mathcal C\).
Thus Condition~(2) is equivalent to
\(\operatorname{rank}(\tilde T)=\operatorname{rank}(T).\)

Since each column of \(\tilde T\) or \(T\) is the projection coefficient
vector \([a_1,a_2,\cdots,a_k]^T\) of one generator, we have
\(\operatorname{rank}(\tilde T)
=
\dim(\operatorname{span}(\Pi_{\mathcal C}(\mathcal{P}))),\)
and
\(\operatorname{rank}(T)
=
\dim(\operatorname{span}(\Pi_{\mathcal C}(\mathcal{P}\cup\mathcal{Q}))).\)
Using \(\operatorname{rank}(\tilde T)=\operatorname{rank}(T)\), one has
\(\dim(\operatorname{span}(\Pi_{\mathcal C}(\mathcal{P})))
=
\dim(\operatorname{span}(\Pi_{\mathcal C}(\mathcal{P}\cup\mathcal{Q}))).\)

However, we find Condition~(2) of Lemma~\ref{lemma2} is not well aligned
with the engineering characteristics of quantum DLAs in \(N\) qubits.
First, computing the commutant subspaces
\(\mathcal{P}'
=
\{X\in\operatorname{End}(\mathcal H)\mid [X,A]=0,\ A\in\mathcal P\}\)
requires solving linear constraints in a Hilbert space \(\mathcal H\), so
the size of the commutant matrices grows exponentially with \(N\). Second,
forming the projection matrices \(\tilde T\) and \(T\) involves evaluating
high-dimensional trace terms of the form
\(\operatorname{Tr}(\rho_\alpha^\dagger iH_\beta)\). Most critically,
Lemma~\ref{lemma2} provides a binary criterion. In engineering practice,
exact algebraic equivalence is often too strict.
A candidate operator set \(\mathcal Q\) might violate this rank condition
but still have a large component aligned with the original
central-projection structure of \(\mathcal P\). This motivates the
following diagnostic quantities.

In this paper, we consider the DLAs generated by Pauli strings in a
general quantum system.
For the diagnostic quantity below, we use the center of the original
commutant $\mathcal P'$ as a fixed reference space,
\[
\mathfrak c_{\mathcal P}:=Z(\mathcal P').
\]
This is different from the center
\(Z((\mathcal P\cup\mathcal Q)')\) appearing in the exact criterion of
Lemma~\ref{lemma2}, which depends on the added set \(\mathcal Q\). By
fixing \(\mathfrak c_{\mathcal P}\), different candidate sets
\(\mathcal Q\) can be compared against the same original central
structure of \(\mathcal P\).
Suppose the center \(\mathfrak c_{\mathcal P}\) is a \(k\)-dimensional
space equipped with a Hilbert--Schmidt orthonormal basis
\(\{\rho_1,\rho_2,\dots,\rho_k\}\). Without loss of generality, we align
the basis such that the subspace spanned by the central projections of
\(\mathcal P\) corresponds to the first \(m\) basis vectors:
\(\operatorname{span}(\Pi_{\mathfrak c_{\mathcal P}}(\mathcal P))
=
\operatorname{span}\{\rho_1,\rho_2,\dots,\rho_m\},\)
with
$m\leq k$ .
Let
\(V_{\mathcal P}
:=
\operatorname{span}(\Pi_{\mathfrak c_{\mathcal P}}(\mathcal P)).\)
We denote \(P_{V_{\mathcal P}}\) as the Hilbert--Schmidt orthogonal
projector onto \(V_{\mathcal P}\).

For candidate operator set
\(\mathcal Q=\{iH_{p+1},\cdots,iH_q\}\), its central projection
coefficient matrix is
\[
M_{\mathcal Q}
=
\left[
v_{p+1}\;v_{p+2}\;\cdots\;v_q
\right],
\quad
v_\beta
=
\left(
\operatorname{Tr}(\rho_1^\dagger iH_\beta),
\ldots,
\operatorname{Tr}(\rho_k^\dagger iH_\beta)
\right)^T ,
\]
for \(\beta=p+1,\ldots,q\).
When \(\mathcal Q\) contains a single added generator,
\(M_{\mathcal Q}\) reduces to one central projection coefficient vector.
When \(\mathcal Q\) contains several added generators, it is a coefficient
matrix, not a single combined projection vector.
Define a quantitative index \textit{projection overlap} for Pauli DLAs as
\begin{align}
O(\mathcal{P},\mathcal{Q})
=
\frac{
\|P_{V_{\mathcal P}}M_{\mathcal Q}\|_F^2
}{
\|M_{\mathcal Q}\|_F^2
},
\end{align}
whenever \(M_{\mathcal Q}\neq0\), where \(\|\cdot\|_F\) denotes the
Frobenius norm. 
In the chosen Hilbert--Schmidt orthonormal basis of
\(\mathfrak c_{\mathcal P}\), 
\(P_{V_{\mathcal P}}\) is represented as an
orthogonal projection matrix acting on the columns of \(M_{\mathcal Q}\).
The denominator measures the total central-projection weight of
\(\mathcal Q\), while the numerator measures the component lying in
\(V_{\mathcal P}\). If \(M_{\mathcal Q}=0\), we set
\(O(\mathcal P,\mathcal Q)=0\) by convention.
This index takes values in \([0,1]\) and quantifies the
degree to which the central projection of \(\mathcal Q\) aligns with the
existing central structure of \(\mathcal P\).
When \(O(\mathcal{P},\mathcal{Q})=1\), the central projection coefficient
matrix of \(\mathcal Q\) is entirely contained within the projection span
determined by \(\Pi_{\mathfrak c_{\mathcal P}}(\mathcal P)\).
Thus, \(O(\mathcal P,\mathcal Q)=1\) indicates full alignment with the
central-projection structure of \(\mathcal P\). Exact DLA invariance is
still determined by the two conditions in Lemma~\ref{lemma2}.

The rank-equality criterion proposed in \cite{zimboras2015symmetry}
provides a binary decision: symmetry is either perfectly preserved or
deemed broken. In practical Hamiltonian engineering, such strict algebraic
equivalence is often too restrictive.
When \(O(\mathcal{P},\mathcal{Q})\in(0,1)\), the operator set
\(\mathcal Q\) contains components orthogonal to the original
central-projection span, implying a measurable degree of deviation from
the original central structure.
This suggests a possible change in the DLA structure, at the level of this
algebraic diagnostic.
By introducing the DLA \textit{percentage change} \(D_c\) alongside the
projection overlap, we offer a diagnostic framework for
variational quantum circuit pruning \cite{sim2021adaptive}:
\begin{align}\label{percentage-change}
D_c
=
\frac{
\dim(\langle \mathcal{P}\cup \mathcal{Q} \rangle_{\mathrm{Lie},\mathbb R})
-
\dim(\langle \mathcal{P} \rangle_{\mathrm{Lie},\mathbb R})
}{
\dim(\langle \mathcal{P} \rangle_{\mathrm{Lie},\mathbb R})
}
\times 100\%.
\end{align}

Specifically, a candidate operator set \(\mathcal Q\) with \(D_c\approx0\)
indicates that it introduces little or no dimension growth relative to the
existing circuit \(\mathcal P\), suggesting algebraic redundancy
in the theory of overparametrization \cite{larocca2023theory}. A high
overlap \(O(\mathcal P,\mathcal Q)\approx1\) provides a quantitative
diagnostic of alignment with the original central-projection structure.

The pair of quantities \(O(\mathcal P,\mathcal Q)\) and \(D_c\) should
therefore be understood as heuristic algebraic diagnostics. They do not by
themselves bound reachable-set distance, channel distance, Trotter error,
trainability, or observable error. Any claim of exact simulability should
be made only after both conditions in Lemma~\ref{lemma2} are checked.

\subsection{DLA reduction by generator filtering}

Question~3 asks how to modify a generating set \(\mathcal A\) to a new
set \(\mathcal A'\) such that the resulting DLA realizes a prescribed
target subalgebra. In this subsection, we answer this question through commutator filtering and Lie-algebraic
projection.
To turn Question~3 into a precise Lie-algebraic statement, we
specify the class of target subalgebras considered in this work. Rather
than treating an arbitrary subalgebra
\(\mathfrak s\subseteq\mathfrak g_{\mathcal A}\), we focus on target
ideals obtained by selecting simple ideals from the reductive
decomposition of \(\mathfrak g_{\mathcal A}\). Detailed proofs of the following results in this subsection are given in Appendix~\ref{proof-of-reduction}.

\subsubsection{Commutator filtering}
To understand why commutator filtering can recover a simple component, we
first consider a simple real Lie algebra \(\mathfrak g\) and an element \(F\in\mathfrak g\).
Lemma~\ref{new-lemma-for-reduction} and Proposition~\ref{prop2} provide the answer needed for compact simple
components.
Throughout this subsection, write $\ad_F:\mathfrak g\to\mathfrak g$, $X\mapsto [F,X]$.
We also write $\im(\ad_F)$ and $\ker(\ad_F)$ for its image and kernel, respectively.

\begin{lemma}\label{new-lemma-for-reduction}
Let $\mathfrak g$ be a finite-dimensional simple Lie algebra over $\mathbb R$, and let $F\in\mathfrak g$.
Suppose that $\mathfrak g=\ker(\ad_F)\oplus\im(\ad_F)$ as real vector spaces, and that $\im(\ad_F)\ne0$.
Then
\(\big\langle [F,X]:X\in\mathfrak g\big\rangle_{\Lie,\mathbb R}
    =\mathfrak g.\)
\end{lemma}

\begin{Proposition}\label{prop2}
Let $\mathfrak g$ be a compact simple real Lie algebra.
If $0\ne F\in\mathfrak g$, then
\(\big\langle [F,X]:X\in\mathfrak g\big\rangle_{\Lie,\mathbb R}
    =\mathfrak g.\)
\end{Proposition}

Compactness provides a positive definite invariant inner product, with respect to which $\ad_F$ is skew-adjoint. Hence the kernel-image decomposition required in Lemma \ref{new-lemma-for-reduction} holds automatically. This leads to Proposition \ref{prop2}.  Proposition \ref{prop2} provides the component-level capacity, but not yet the finite-generator reduction.

Let
\[
    \mathfrak g_{\mathcal A}
    =
    Z(\mathfrak g_{\mathcal A})
    \oplus
    \bigoplus_{j=1}^{T}\mathfrak g_j
\]
be a reductive dynamical Lie algebra, where $Z(\mathfrak g_{\mathcal A})$ is the center and the $\mathfrak g_j$ are simple ideals.
Let $\mathcal A=\{A_1,\ldots,A_L\}$ generate $\mathfrak g_{\mathcal A}$, and let 
\[\mathfrak s=\bigoplus_{j\in S}\mathfrak g_j\]
be the target semisimple ideal, where
\(S\subseteq\{1,\ldots,T\}\).
Write $\pi_j:\mathfrak g_{\mathcal A}\to\mathfrak g_j$ for the projection.
For $F\in\mathfrak s$, put $F_j=\pi_j(F)$, and define 
\[\mathcal A'_F:=\{[F,A_l]:1\le l\le L\}, \quad \mathfrak g_{\mathcal A'_F}:=\langle\mathcal A'_F\rangle_{\Lie,\mathbb R}.\] 

We now return to DLA reduction. Before proving that the filtered
generators recover the whole target ideal, we first need a containment
result: the filtered algebra should not leave the target \(\mathfrak s\). Since
\(\mathfrak s\) is a direct sum of selected simple ideals in the
reductive decomposition, it is itself an ideal of \(\mathfrak g_{\mathcal A}\).
Therefore, if \(F\in\mathfrak s\), all commutators
\([F,A_l]\) remain inside \(\mathfrak s\). This is Proposition~\ref{prop3}; its proof follows directly from the fact
that \(\mathfrak s\) is an ideal of \(\mathfrak g_{\mathcal A}\).

\begin{Proposition}\label{prop3}
For every $F\in\mathfrak s$, one has $\mathfrak g_{\mathcal A'_F}\subseteq\mathfrak s$.
In particular, if $\dim\mathfrak g_{\mathcal A'_F}=\dim\mathfrak s$, then $\mathfrak g_{\mathcal A'_F}=\mathfrak s$.
\end{Proposition}

In a direct sum of simple Lie algebras, surjectivity onto each individual
factor is not enough to guarantee that a subalgebra is the full direct
sum. The standard obstruction is a diagonal copy, such as
\(\{(X,X):X\in\mathfrak{su}(2)\}
    \subset
    \mathfrak{su}(2)\oplus\mathfrak{su}(2),\)
which projects surjectively onto both factors but is only one copy of
\(\mathfrak{su}(2)\). Therefore, we need to rule
out such diagonal locking.

\begin{lemma}\label{lemma5}
Let $\mathfrak s_1,\ldots,\mathfrak s_m$ be finite-dimensional simple real Lie algebras, and let $\mathfrak q\le\bigoplus_{i=1}^{m}\mathfrak s_i$.
Assume that every coordinate projection of $\mathfrak q$ is surjective.
Assume also that, whenever $i\ne j$ and $\mathfrak s_i\cong\mathfrak s_j$, the projection of $\mathfrak q$ onto $\mathfrak s_i\oplus\mathfrak s_j$ is surjective.
Then
\[
    \mathfrak q=\bigoplus_{i=1}^{m}\mathfrak s_i.
\]
\end{lemma}

\begin{Proposition}\label{prop4}
Let $F\in\mathfrak s$.
Assume that
\(\big\langle [F_j,\pi_j(A_l)]:1\le l\le L\big\rangle_{\Lie,\mathbb R}
    =\mathfrak g_j\)
for every $j\in S$.
Assume moreover that, for every distinct $j,k\in S$ with $\mathfrak g_j\cong\mathfrak g_k$, one has
\[
    \big\langle
    \big([F_j,\pi_j(A_l)],[F_k,\pi_k(A_l)]\big):1\le l\le L
    \big\rangle_{\Lie,\mathbb R}
    =
    \mathfrak g_j\oplus\mathfrak g_k.
\]
Then $\mathfrak g_{\mathcal A'_F}=\mathfrak s$.
\end{Proposition}

For a given candidate \(F\), Proposition~\ref{prop4} provides Lie-closure checks on each
selected simple component and on each pair of isomorphic components,
thereby certifying whether the filtered generators
recover the target ideal \(\mathfrak s\). 

\subsubsection{A Lie algebra retraction}
The map $X\mapsto [F,X]$ is a derivation, not a Lie algebra homomorphism.
Indeed,
\(\ad_F([U,V])=[\ad_F(U),V]+[U,\ad_F(V)],\)
whereas, in general, $\ad_F([U,V])\ne [\ad_F(U),\ad_F(V)]$.
This suggests that Lie generation is not automatically maintained when passing to the first-order filter. A Lie algebra retraction could serve as a good alternative approach.

\begin{Proposition}\label{prop5}
Let $\mathfrak g_{\mathcal A}=\langle A_1,\ldots,A_L\rangle_{\Lie,\mathbb R}$ be a finite-dimensional real Lie algebra, and let $\mathfrak s$ be a Lie subalgebra.
Suppose that there is a Lie algebra homomorphism $\rho:\mathfrak g_{\mathcal A}\to\mathfrak s$ whose restriction to $\mathfrak s$ is the identity.
Then
\(\big\langle \rho(A_l):1\le l\le L\big\rangle_{\Lie,\mathbb R}
    =\mathfrak s.\)
\end{Proposition}

We apply this to the reductive situation.
Let
\[
    \mathfrak g_{\mathcal A}
    =
    Z(\mathfrak g_{\mathcal A})
    \oplus
    \bigoplus_{j=1}^{T}\mathfrak g_j
\]
be a reductive real Lie algebra, where $Z(\mathfrak g_{\mathcal A})$ is the center and the $\mathfrak g_j$ are simple ideals.
Fix $S\subseteq\{1,\ldots,T\}$, a subspace $\mathfrak z_\mathfrak s$ of $Z(\mathfrak g_{\mathcal A})$, and a vector-space complement $Z(\mathfrak g_{\mathcal A})=\mathfrak z_\mathfrak s\oplus\mathfrak z_\mathfrak k$.
Put
\[
    \mathfrak s
    =
    \mathfrak z_\mathfrak s\oplus\bigoplus_{j\in S}\mathfrak g_j,
    \quad
    \mathfrak k
    =
    \mathfrak z_\mathfrak k\oplus\bigoplus_{j\notin S}\mathfrak g_j.
\]
Then $\mathfrak g_{\mathcal A}=\mathfrak s\oplus\mathfrak k$.
Let $\pi_{\mathfrak s}:\mathfrak g_{\mathcal A}\to\mathfrak s$ be the projection along $\mathfrak k$.

Proposition \ref{prop5} shows that if a Lie algebra retraction exists, DLA reduction follows automatically. The next question is whether such a retraction naturally exists in the reductive decomposition.

\begin{theorem}\label{thm-subalgebra-modify}
If $A_1,\ldots,A_L$ generate $\mathfrak g_{\mathcal A}$, then $\pi_{\mathfrak s}$ is a Lie algebra homomorphism, restricts to the identity on $\mathfrak s$, and
\(\big\langle \pi_{\mathfrak s}(A_l):1\le l\le L\big\rangle_{\Lie,\mathbb R}
    =\mathfrak s.\)
\end{theorem}

Theorem~\ref{thm-subalgebra-modify} provides a projection-based reduction with
an algebraic guarantee. This contrasts with commutator filtering: projection is
a Lie algebra homomorphism and therefore preserves Lie generation automatically,
whereas \(F\)-filtering is only a derivation and must be verified using the
conditions in Proposition~\ref{prop4}.

Question~3 asks how one can modify a generating set \(\mathcal A\) to a new set
\(\mathcal A'\) such that the resulting DLA realizes \(\mathfrak s\subseteq\mathfrak g_{\mathcal A}\). The results above
provide a rigorous version of this question in the compact reductive setting.
The reduction result has a direct resource interpretation. 
If the target dynamics is contained in the dynamical Lie subgroup
\(\exp(s)\), then either the projected generating set
\(\{\pi_s(A_l):1\le l\le L\}\) from Theorem~\ref{thm-subalgebra-modify}, or a filtered generating set
satisfying Proposition~\ref{prop4}, realizes this target DLA.

In particular, when
\(\dim\mathfrak s<\dim\mathfrak g_{\mathcal A}\), the effective
dynamical search space is reduced from \(\dim\mathfrak g_{\mathcal A}\)
to \(\dim\mathfrak s\). This provides a Lie-algebraic notion of resource
reduction: the task is implemented within the smaller reachable group, rather than within the full group.

We finish this subsection with two elementary examples. The first shows that \(F\) must be
chosen relative to the finite generating set \(\mathcal A\); the second
shows that pairwise conditions are needed to rule out diagonal locking
between isomorphic simple ideals.

Example 1 (Dependence on the finite generating set):
We illustrate why a nonzero component \(F_j\in\mathfrak g_j\) is not sufficient for finite-generator filtering. Let
\(A_1=iX_1+2iX_2,\) and 
\(A_2=2iY_1+iY_2,\)
so that
\(\mathfrak g_{\mathcal A}
    =
    \mathfrak g_1\oplus\mathfrak g_2,\)
with \(\mathfrak g_1,\mathfrak g_2\cong\mathfrak{su}(2).\)   
Take the target ideal
\(\mathfrak s=\mathfrak g_1.\)
If one chooses
\(F=iX_1,\)
then, up to nonzero real scalar factors,
\([F,A_1]=0,\) and 
    $[F,A_2]\propto iZ_1$.
Hence
\(\mathfrak g_{\mathcal A'_F}
    =
    \operatorname{span}_{\mathbb R}\{iZ_1\},\)
which is strictly smaller than \(\mathfrak g_1\cong\mathfrak{su}(2)\).
In contrast, choosing
\(F=iZ_1\)
gives
\( [F,A_1]\propto iY_1,\) and 
$[F,A_2]\propto iX_1$.
Thus the filtered generators recover
\(\langle iX_1,iY_1\rangle_{\mathrm{Lie},\mathbb R}
    =
    \operatorname{span}_{\mathbb R}\{iX_1,iY_1,iZ_1\}
    =
    \mathfrak g_1.\)

Example 2 (Avoiding diagonal locking between isomorphic ideals):
Let
\(\mathfrak g_{\mathcal A}
    =
    \mathfrak{su}(2)_1\oplus\mathfrak{su}(2)_2\)
be generated by
\[
    A_1=iX_1+iX_2,
    \quad
    A_2=iY_1+iY_2,
    \quad
    A_3=iZ_1+2iZ_2.
\]
The set \(\{A_1,A_2,A_3\}\) generates the full direct sum, since
\([A_1,A_2]\) gives a direction proportional to \(iZ_1+iZ_2\), which can
be combined with \(A_3=iZ_1+2iZ_2\) to separate the two \(Z\)-directions;
the remaining \(X\)- and \(Y\)-directions are then obtained by
commutators.

Now take the full target
\(\mathfrak s
    =
    \mathfrak{su}(2)_1\oplus\mathfrak{su}(2)_2.\)
If one chooses the symmetric filtering operator
\(F=iZ_1+iZ_2,\)
then
\[
    [F,A_1]\propto iY_1+iY_2,
    \quad
    [F,A_2]\propto iX_1+iX_2,
    \quad
    [F,A_3]=0.
\]
Therefore the filtered algebra is only the diagonal copy
\(\operatorname{span}_{\mathbb R}
    \{iX_1+iX_2,\ iY_1+iY_2,\ iZ_1+iZ_2\}
    \cong
    \mathfrak{su}(2),\)
rather than
\(\mathfrak{su}(2)_1\oplus\mathfrak{su}(2)_2\).

In contrast, choosing
\(F=iZ_1+2iZ_2\)
breaks the diagonal locking. Indeed,
\([F,A_1]\propto iY_1+2iY_2,
    [F,A_2]\propto iX_1+2iX_2.\)
Their commutator yields a direction proportional to
\(iZ_1+4iZ_2.\)
Further commutators produce, for example,
\(iY_1+8iY_2,\)
which can be linearly combined with \(iY_1+2iY_2\) to separate
\(iY_1\) and \(iY_2\). The remaining \(iX_j\) and \(iZ_j\) directions are
then recovered by commutators. Consequently,
\(\mathfrak g_{\mathcal A'_F}
    =
    \mathfrak{su}(2)_1\oplus\mathfrak{su}(2)_2
    =
    \mathfrak s.\)

The two examples illustrate how the candidate \(F\) can be checked in
practice. In general, the computational cost of this verification depends on
the representation of the generators, the method used to obtain the reductive
decomposition, and the dimension of the Lie algebra generated during the
calculation. Thus, the verification should be understood as a structural check whose cost depends on the specific Hamiltonian generators and on the Lie algebra generated by them.

\section{Applications and Numerical Studies}

\subsection{Computational methodology}

The numerical Lie-algebra computations in this section are performed in
finite-dimensional matrix representations. We use the standard Pauli matrices,
and DLA generators are represented in the skew-Hermitian convention \(iP\).
Matrix Lie algebras are treated as real vector spaces, with numerical linear
independence tested using the real Hilbert--Schmidt inner product.

Lie closures are computed iteratively from the real linear span of the input
generators by adjoining commutators that increase the numerical rank. The
iteration terminates when no additional independent commutator is found.

For a matrix set \(\mathcal S=\{S_1,\ldots,S_r\}\), the commutant
\(\mathcal S'\) is obtained from the homogeneous system
\([X,S_j]=0\) for \(j=1,\ldots,r\). The center \(Z(\mathcal S')\) is then
obtained by restricting a general linear combination of commutant basis
elements to commute with the full commutant basis. Matrix ranks are determined
from the singular values. Further implementation details and numerical
tolerances are given in Appendix~\ref{app:numerical-procedures}.

\subsection{DLA composition for Hilbert-space allocation}

Consider a \(2\)-qubit quantum system with a dynamical generating set
\(\mathcal{A}=\{A_1,A_2\}\) describing a dipole coupling with a tilted
magnetic field 
\(A_1=i\left(2\sigma^z_1\sigma^z_2-\sigma^x_1\sigma^x_2-\sigma^y_1\sigma^y_2\right)\)
and
\(A_2=i\left(\sigma^x_1-\sigma^y_1+\sigma^x_2-\sigma^y_2\right)\) \cite{zimboras2015symmetry}.
Define another generating set
\[
\mathcal{B}=\{B_1\}
=
\left\{
i\left(\sigma^x_1\sigma^x_2+\sigma^y_1\sigma^y_2+\sigma^z_1\sigma^z_2\right)
\right\}
\]
corresponding to a Heisenberg-type interaction. Here,
\(\sigma^x_j,\sigma^y_j,\sigma^z_j\) denote Pauli operators acting on
qubit \(j\), and all \(A_1,A_2,B_1\) are traceless anti-Hermitian
generators acting on the same two-qubit system Hilbert space.

\begin{figure}[!t]
\centering
\subfloat[\textbf{In DLA}]{
    \includegraphics[width=0.6\textwidth]{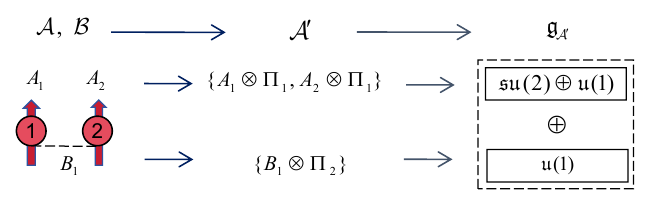}
    \label{fig:res-classical}
}
\hfill
\subfloat[\textbf{In circuit design}]{
   \begin{quantikz}[row sep=0.13cm, column sep=0.13cm]
    \lstick{$q_1$ } & \gate[2]{U_{A_1}} & \gate[2]{U_{A_2}} & \qw & \qw \\
    \lstick{$q_2$} &                   &                   & \qw & \qw \\
   \lstick{$q_3$ } & \qw               & \qw               & \gate[2]{U_{B_1}} & \qw \\
    \lstick{$q_4$} & \qw               & \qw               &                   & \qw
\end{quantikz} \quad $\rightarrow$ \quad 
    \begin{quantikz}[row sep=0.3cm, column sep=0.3cm]
    \lstick{$q_{\text{aux}}$}    & \gate{H} & \ctrl[open]{1} & \ctrl[open]{1}  & \ctrl{1} & \gate{H} &  \\
    \lstick{$q_1$}    &                       &\gate[2, style={fill=blue!20}]{U_{A_1} \otimes \Pi_1} & \gate[2, style={fill=blue!20}]{U_{A_2} \otimes \Pi_1} & \gate[2, style={fill=red!20}]{U_{B_1} \otimes \Pi_2} & & \\
    \lstick{$q_2$}    &                       &                           &                         &                         & &
    \end{quantikz}
    \label{fig:res-quantum}
}
\caption{Illustration of DLA composition.
(a) DLA composition using the dipole--field and
Heisenberg-type generator sets. (b) An example of evolution using circuits. If \(A_1\) is anti-Hermitian, then
\(e^{(A_1\otimes\Pi_1)t}
=
e^{A_1t}\otimes\Pi_1+I\otimes\Pi_2 .\)
This construction uses \(2+\lceil\log_2 2\rceil=3\) qubits for
Hilbert-space allocation, compared with $4$ qubits in naive implementation. 
 }
\label{example1}
\end{figure}

Let
\(\Pi_1=\ket{0}\bra{0},
\Pi_2=\ket{1}\bra{1},\)
and \(\chi=\Pi_1-\Pi_2.\)
According to Theorem~\ref{composition1}, a new generating set is
\begin{align}\label{example1-a'}
\mathcal{A}'
=
\left\{
A_1\otimes\Pi_1,\,
A_2\otimes\Pi_1,\,
B_1\otimes\Pi_2
\right\}.
\end{align}
For any \(A\in\mathcal{A}\) and \(B\in\mathcal{B}\),
\([A\otimes\Pi_1,B\otimes\Pi_2]
=
AB\otimes(\Pi_1\Pi_2)-BA\otimes(\Pi_2\Pi_1)
=
0,\)
showing that the subalgebras generated by the two projector-labeled
branches commute trivially. The subset
\(\{A_1\otimes\Pi_1,A_2\otimes\Pi_1\}\) generates a subalgebra isomorphic
to \(\mathfrak g_{\mathcal A}\), since its commutators preserve the
structure of \(\mathfrak g_{\mathcal A}\),
\[
[A_i\otimes\Pi_1,A_j\otimes\Pi_1]
=
[A_i,A_j]\otimes\Pi_1
\in
\mathfrak g_{\mathcal A}\otimes\Pi_1.
\]
Similarly, \(\{B_1\otimes\Pi_2\}\) generates a subalgebra isomorphic to
\(\mathfrak g_{\mathcal B}\), represented on the second auxiliary branch
as \(\mathfrak g_{\mathcal B}\otimes\Pi_2\). Therefore, the new generating
set \(\mathcal A'\) yields
\[
\mathfrak g_{\mathcal A'}
\cong
\mathfrak g_{\mathcal A}\oplus\mathfrak g_{\mathcal B}
\cong
\left(\mathfrak{su}(2)\oplus\mathfrak{u}(1)\right)\oplus\mathfrak{u}(1).
\]
This example uses one auxiliary label qubit, since the two branches are
encoded by the rank-one projectors \(\Pi_1\) and \(\Pi_2\). This
qubit-count statement concerns Hilbert-space allocation and register
count. The displayed construction uses
\(|\mathcal A'|=|\mathcal A|+|\mathcal B|=3\)
controlled generators.
An illustration of this application is given in Fig.~\ref{example1}.

\subsection{DLA invariance diagnostics for the central-spin model}

We illustrate the two diagnostics using a central-spin-type example embedded
in a three-qubit Hilbert space \cite{merkulov2002electron}. The purpose of this example is to compare the projection overlap
\(O(\mathcal P,\mathcal Q)\) with the DLA percentage change \(D_c\).
These quantities are used as algebraic diagnostics and are not intended
to provide an operational simulation error bound.

Define \(\mathcal P=\{iH_1,iH_2\}\), where
\(iH_1=i\{0.5\sigma_1^z+
J_2(\sigma_1^x\sigma_2^x+\sigma_1^y\sigma_2^y)\}\) and
\(iH_2=i\sigma_1^z\), with \(J_2\in\mathbb R\setminus\{0\}\).
All operators are understood as three-qubit operators, with identity
tensor factors on unused qubits suppressed.
We now compute the fixed reference center for this \(\mathcal P\). Let
\(\Lambda_+=|00\rangle\langle00|\otimes I_2\),
\(\Lambda_0=(|01\rangle\langle01|+|10\rangle\langle10|)\otimes I_2\),
and \(\Lambda_-=|11\rangle\langle11|\otimes I_2\). The generators in
\(\mathcal P\) preserve the block decomposition determined by
\(\Lambda_+,\Lambda_0,\Lambda_-\), and act irreducibly on the
two-dimensional \(\Lambda_0\) block, up to the identity factor on the
third qubit. Hence the traceless fixed reference center is
\(\mathfrak c_{\mathcal P}:=Z(\mathcal P')\cap\mathfrak{su}(8)
=\operatorname{span}_{\mathbb R}\{\rho_1,\rho_2\}\), where we choose
\(\rho_1=\frac{i}{2}(\Lambda_+-\Lambda_-)
=\frac{i}{4}(\sigma_1^z+\sigma_2^z)\) and
\(\rho_2=\frac{i}{2\sqrt2}(\Lambda_+-\Lambda_0+\Lambda_-)
=\frac{i}{\sqrt8}\sigma_1^z\sigma_2^z\). Both
\(\rho_1\) and \(\rho_2\) belong to \(\mathfrak{su}(8)\), and
\(\operatorname{Tr}(\rho_\alpha^\dagger\rho_\beta)=\delta_{\alpha\beta}\).

Let \(\Pi_{\mathcal P}\) be the Hilbert--Schmidt orthogonal projection
onto \(\mathfrak c_{\mathcal P}\). For an operator \(X\), write
\(\Pi_{\mathcal P}(X)=v_{\mathcal P,1}(X)\rho_1+
v_{\mathcal P,2}(X)\rho_2\), and denote
\(v_{\mathcal P}(X)=(v_{\mathcal P,1}(X),v_{\mathcal P,2}(X))^T\).
With the basis above, \(V_{\mathcal P}:=\operatorname{span}(\Pi_{\mathcal P}(\mathcal P))
=\operatorname{span}\{\rho_1\}\). The flip-flop term
\(\sigma_1^x\sigma_2^x+\sigma_1^y\sigma_2^y\) is Hilbert--Schmidt
orthogonal to \(\rho_1\) and \(\rho_2\), while
\(\Pi_{\mathcal P}(i\sigma_1^z)=2\rho_1\). Hence
\(\Pi_{\mathcal P}(iH_1)=\rho_1\) and
\(\Pi_{\mathcal P}(iH_2)=2\rho_1\), equivalently
\(v_{\mathcal P}(iH_1)=(1,0)^T\) and
\(v_{\mathcal P}(iH_2)=(2,0)^T\). Therefore, the
fixed-reference projection coefficient matrix of \(\mathcal P\), whose
columns are the coefficient vectors of the projected generators, is
\[
\tilde T
=
\begin{bmatrix}
1 & 2\\
0 & 0
\end{bmatrix},
\quad
\operatorname{rank}(\tilde T)=1.
\]

We first choose an added generator aligned with the original projection
span. Let \(\mathcal Q_1=\{iH_3\}\), where
\(iH_3=i\{0.25\sigma_1^z+
J_3(\sigma_1^x\sigma_2^x+\sigma_1^y\sigma_2^y)\}\), with
\(J_3\in\mathbb R\). Since the flip-flop term has zero projection onto
\(\mathfrak c_{\mathcal P}\), we have
\(\Pi_{\mathcal P}(iH_3)=0.5\rho_1\). Hence
\(O(\mathcal P,\mathcal Q_1)=1\). Moreover, since
\(iH_3=(J_3/J_2)iH_1+(0.25-0.5J_3/J_2)iH_2\), the added generator lies
in the linear span of \(\mathcal P\), and the Lie closure is unchanged.
Thus \(D_c=0\) for \(\mathcal Q_1\).

Next introduce the concrete traceless perturbation direction
\(R=\rho_2\), namely
\(R=\frac{i}{2\sqrt2}(\Lambda_+-\Lambda_0+\Lambda_-)
=\frac{i}{\sqrt8}\sigma_1^z\sigma_2^z\). This direction is
Hilbert--Schmidt normalized, belongs to \(\mathfrak{su}(8)\), and
satisfies \(v_{\mathcal P}(R)=(0,1)^T\). It is therefore a nontrivial traceless direction in the fixed
reference center.

\begin{figure}[t!]
\centering
\includegraphics[width=8cm]{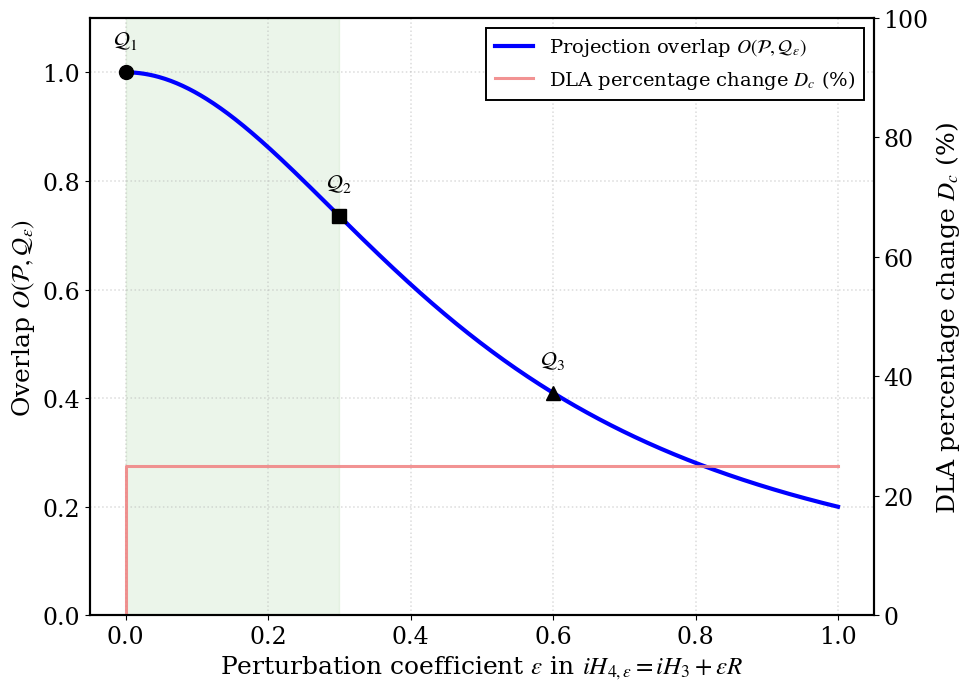}
\caption{DLA invariance diagnostics for the central-spin model. The
horizontal axis is the perturbation coefficient \(\epsilon\) in
\(\mathcal Q_\epsilon=\{iH_{4,\epsilon}\}\), where
\(iH_{4,\epsilon}=iH_3+\epsilon R\) and
\(R=i\sigma_1^z\sigma_2^z/\sqrt8\). The blue curve shows
\(O(\mathcal P,\mathcal Q_\epsilon)\), and the red curve shows the DLA
percentage change \(D_c\). The shaded region marks the high-overlap
interval containing the representative point \(\mathcal Q_2\).}
\label{center-spin}
\end{figure}

For \(\epsilon\in[0,1]\), define
\(iH_{4,\epsilon}=iH_3+\epsilon R\) and
\(\mathcal Q_\epsilon=\{iH_{4,\epsilon}\}\). Then
\(v_{\mathcal P}(iH_{4,\epsilon})=(0.5,\epsilon)^T\). In
the notation of the preceding subsection, \(M_{\mathcal Q_\epsilon}\) is
this single coefficient vector. Since \(P_{V_{\mathcal P}}\) projects
onto the first coordinate, the projection overlap is
\(O(\mathcal P,\mathcal Q_\epsilon)=0.5^2/(0.5^2+\epsilon^2)\), which is
the blue curve in Fig.~\ref{center-spin}. For \(\epsilon=0\), this family
reduces to the aligned case \(\mathcal Q_1\). For any \(\epsilon>0\), the
added generator has a nonzero component along \(\rho_2\), and the
fixed-reference rank-equality diagnostic is no longer satisfied:
\[
T_\epsilon
=
\begin{bmatrix}
1 & 2 & 0.5\\
0 & 0 & \epsilon
\end{bmatrix},
\quad
\operatorname{rank}(T_\epsilon)=2.
\]
This illustrates a practical limitation of using the exact invariance
criterion alone. Lemma~\ref{lemma2} determines whether the Lie closure is
exactly preserved, but it does not provide a graded measure of how the
added generator departs from the original projected structure. The overlap
\(O(\mathcal P,\mathcal Q_\epsilon)\) is therefore used as an
algebraic diagnostic of alignment in the fixed reference center, whereas
\(D_c\) records the actual dimension change of the generated Lie algebra.

For the generator sets above,
\(\dim(\langle\mathcal P\rangle_{\mathrm{Lie},\mathbb R})=4\). At
\(\epsilon=0\), \(\mathcal Q_\epsilon=\mathcal Q_1\), so the Lie closure
is unchanged and \(D_c=0\). For every \(\epsilon>0\), the added traceless
direction \(R\) increases the Lie-closure dimension to \(5\), giving
\(D_c=25\%\). Thus the red curve in Fig.~\ref{center-spin} records the
appearance of an additional Lie-algebraic direction, while the blue curve
varies continuously with \(\epsilon\).

The marked points in Fig.~\ref{center-spin} correspond to three
representative values. The point \(\mathcal Q_1\) corresponds to
\(\epsilon=0\), so \(O(\mathcal P,\mathcal Q_1)=1\) and \(D_c=0\). The
point \(\mathcal Q_2\) corresponds to \(\epsilon=0.3\), for which
\[
O(\mathcal P,\mathcal Q_2)
=
\frac{0.5^2}{0.5^2+0.3^2}
\approx0.735.
\]
The point \(\mathcal Q_3\) corresponds to \(\epsilon=0.6\), for which
\[
O(\mathcal P,\mathcal Q_3)
=
\frac{0.5^2}{0.5^2+0.6^2}
\approx0.410.
\]

The shaded region in Fig.~\ref{center-spin} marks an illustrative
high-overlap range. For any nonzero \(\epsilon\), the fixed-reference
projection already contains an additional coefficient direction, whereas
the overlap can remain large when \(\epsilon\) is small. This illustrates
that Lemma~\ref{lemma2} is an exact invariance criterion, while
\(O(\mathcal P,\mathcal Q_\epsilon)\) and \(D_c\) provide complementary
algebraic information: the former captures the remaining fixed-center
alignment, and the latter identifies the resulting change in the generated
Lie-algebra dimension.

\subsection{DLA reduction for trainability}
\label{subsec:dla-score-bp}

We use DLA reduction as an algebraic pruning principle for trainability.
Our motivation is the Lie algebraic theory of barren plateaus, where the
variance scale of sufficiently deep circuits contains the dimension factor
\(1/\dim(\mathfrak g)\)~\cite{ragone2024lie}. Reducing the DLA dimension can
therefore improve this dimension based proxy, but only if the reduction retains
the Hamiltonian components relevant to the task. We achieve this by selecting
the simple DLA components with the largest projection weights of \(iH\).

Consider an \(n\)-qubit system divided into \(B\) blocks of two qubits, so that
\(n=2B\). Let \(\mathcal B_j\) denote the \(j\)-th block, with Pauli operators
\(X_{j,a},Y_{j,a},Z_{j,a}\) for \(a=1,2\). 
The Hamiltonian on block
\(\mathcal B_j\) is
\begin{equation}\label{H_j}
    H_j
    =
    J_j Z_{j,1}Z_{j,2}
    +
    h_{x,j}(X_{j,1}+X_{j,2})
    +
    h_{z,j}(Z_{j,1}+Z_{j,2}).
\end{equation}
A Hamiltonian instance is then defined by
\begin{equation}\label{eq:random-block-hamiltonian}
    H=\sum_{j=1}^{B}w_jH_j .
\end{equation}
For each instance, the positive block weights are sampled as
\(\tilde w_j\sim\operatorname{LogNormal}(0,1)\) and rescaled according to
\(w_j=\tilde w_j/\max_{1\leq k\leq B}\tilde w_k\), so that the largest weight
is equal to one.

The variational loss is
\begin{equation}
    \mathcal L(\theta)
    =
    \operatorname{Tr}\!\left[
    H U(\theta)\rho_0 U^\dagger(\theta)
    \right],
    \label{eq:random-hamiltonian-loss}
\end{equation}
where \(\rho_0\) is the input state and \(U(\theta)\) is generated by a chosen
Hamiltonian generator set. For each block we use the local generator pool
\(\mathcal A_j=
\{iX_{j,1},iY_{j,1},iX_{j,2},iY_{j,2},iY_{j,1}Y_{j,2}\}\).
It generates
\(\langle\mathcal A_j\rangle_{\mathrm{Lie},\mathbb R}
\cong\mathfrak{su}(4)\). With
\(\mathcal A=\bigcup_{j=1}^{B}\mathcal A_j\), different blocks commute, and
hence
\begin{equation}\nonumber
    \mathfrak g_{\mathcal A}
    =
    \bigoplus_{j=1}^{B}\mathfrak g_j,
    \quad
    \mathfrak g_j\cong\mathfrak{su}(4).
\end{equation}
Thus, \(\dim(\mathfrak g_{\mathcal A})=15B\).

We assign a DLA score to each block by projecting the task Hamiltonian onto the
corresponding simple component. Let
\(\pi_j:\mathfrak g_{\mathcal A}\to\mathfrak g_j\) be the \(j\)-th projection.
Using the normalized Hilbert--Schmidt norm
\(\|X\|_{\mathrm{HS}}^2=2^{-n}\operatorname{Tr}(X^\dagger X)\), we define
\begin{equation}
 R_j=\|\pi_j(iH)\|_{\mathrm{HS}}^2 .
\label{eq:dla-relevance-score}
\end{equation}
For the Hamiltonian in Eq.~\eqref{eq:random-block-hamiltonian}, orthogonality of
distinct Pauli strings gives
\(R_j=w_j^2(J_j^2+2h_{x,j}^2+2h_{z,j}^2)\). Thus, \(R_j\) is the normalized
squared Hilbert--Schmidt weight of $H_j$ supported on
\(\mathcal B_j\).

We select \(S\) blocks with the largest DLA scores and define
\(\mathcal I_{\mathrm{DLA}}=
\operatorname{Top}_{S}\{R_j\}_{j=1}^{B}\). The target ideal is
\begin{equation}\nonumber
    \mathfrak s
    =
    \bigoplus_{j\in\mathcal I_{\mathrm{DLA}}}\mathfrak g_j,
    \quad
    \dim(\mathfrak s)=15S .
\end{equation}
This target algebra is obtained by selecting simple ideals from the
reductive decomposition of \(\mathfrak g_{\mathcal A}\).

The reduced generator set is obtained through the projection form of DLA
reduction. Let \(\mathfrak g_{\mathcal A}=\mathfrak s\oplus\mathfrak k\), where
\(\mathfrak s=\bigoplus_{j\in\mathcal I_{\mathrm{DLA}}}\mathfrak g_j\) and
\(\mathfrak k=\bigoplus_{j\notin\mathcal I_{\mathrm{DLA}}}\mathfrak g_j\) are
complementary ideals. Denote by
\(\pi_{\mathfrak s}:\mathfrak g_{\mathcal A}\to\mathfrak s\) the projection
onto \(\mathfrak s\) along \(\mathfrak k\), and define
\(\mathcal A_{\mathrm{red}}
    =
    \{\pi_{\mathfrak s}(A):A\in\mathcal A\}\setminus\{0\}.\)
For the block decomposition above, this reduces to
\(\mathcal A_{\mathrm{red}}=
\bigcup_{j\in\mathcal I_{\mathrm{DLA}}}\mathcal A_j\). Since
\(\pi_{\mathfrak s}\) is a Lie algebra retraction,
Theorem~\ref{thm-subalgebra-modify} gives
\begin{equation}
    \left\langle
    \mathcal A_{\mathrm{red}}
    \right\rangle_{\mathrm{Lie},\mathbb R}
    =
    \mathfrak s.
    \label{eq:reduced-dla-equality}
\end{equation}
Thus, the reduced generators realize exactly the selected ideal
\(\mathfrak s\), preserving the full local algebra
\(\mathfrak{su}(4)\) on each retained block while reducing the DLA dimension
from \(15B\) to \(15S\). 

We compare five ansatz choices. The
full ansatz serves as a global baseline: it uses all block-local generators
together with nearest-block \(ZZ\) couplings, and its dimension factor in
Fig.~5(a) is evaluated using the full-algebra proxy \(1/(4^n-1)\). The block
local ansatz retains all \(B\) components and generates
\(\bigoplus_{j=1}^B\mathfrak{su}(4)\), but contains no couplings between
distinct blocks. We further consider naive spatial pruning, random pruning,
and DLA score pruning, each of which retains \(S\) block components. Naive
spatial pruning retains the first \(S\) blocks, random pruning samples \(S\)
blocks uniformly without replacement from \(\{1,\ldots,B\}\), and DLA score
pruning retains the \(S\) blocks with the largest \(R_j\).

For a block based method \(\mathsf M\), let \(\mathcal I_{\mathsf M}\) denote
its selected block set and define
\(\mathfrak s_{\mathsf M}=
\bigoplus_{j\in\mathcal I_{\mathsf M}}\mathfrak g_j\). Since the block
components are orthogonal under the normalized Hilbert--Schmidt inner product,
we define
\begin{equation}
    C_{\mathsf M}
    =
    \frac{
        \sum_{j\in\mathcal I_{\mathsf M}}R_j
    }{
        \sum_{j=1}^{B}R_j
    }.
    \label{eq:coverage}
\end{equation}
Thus, \(C_{\mathsf M}\) is the fraction of the normalized squared
Hilbert--Schmidt norm of \(iH\) retained by the projection onto
\(\mathfrak s_{\mathsf M}\). For methods that retain \(S\) blocks,
\(\dim(\mathfrak s_{\mathsf M})=15S\). A useful pruning rule should therefore
retain a large \(C_{\mathsf M}\) while keeping \(S\), and hence the reduced DLA dimension, small.
To assess whether this reduction also improves the gradient scale, we compute
the sampled gradient variance
\begin{equation}
    V_{\mathsf M}(B)
    =
    \mathbb E_{\mathrm{inst}}
    \left[
        \operatorname{Var}_{\theta,\mu}
        \left(
            \partial_{\theta_\mu}\mathcal L_{\mathsf M}(\theta)
        \right)
    \right].
    \label{eq:sampled-gradient-variance}
\end{equation}
 A barren plateau is
characterized by a gradient variance that decreases exponentially with the
system size~\cite{ragone2024lie}. We therefore use \(V_{\mathsf M}(B)\) to
track the gradient scale as the number of blocks increases. This quantity is a
numerical diagnostic and does not establish the presence or absence
of a barren plateau.

In the numerical implementation, the local coefficients in
Eq.~\eqref{H_j} are sampled independently as
\(J_j=s_{J,j}u_{J,j}\),
\(h_{x,j}=0.7s_{x,j}u_{x,j}\), and
\(h_{z,j}=0.3s_{z,j}u_{z,j}\), where
\(u_{J,j},u_{x,j},u_{z,j}\) are independent samples from the uniform
distribution on \([0.7,1.3]\), and
\(s_{J,j},s_{x,j},s_{z,j}\in\{-1,1\}\) are independent random signs.
We use \(B\in\{2,3,4,5\}\), \(S=2\), and
\(\rho_0=(|0\rangle\!\langle0|)^{\otimes n}\). Each ansatz contains two
layers, with every listed gate applied once per layer. Each gate has the form
\(\exp(-i\theta P/2)\), with parameters sampled uniformly from
\([0,2\pi)\), and gradients are evaluated using the parameter shift rule.
For each Hamiltonian instance, the gradient variance is estimated from 30
parameter samples and at most 12 uniformly sampled parameter indices per
sample. For random pruning, the value for each Hamiltonian instance is first
averaged over five independently sampled block sets. The reported results are
then averaged over eight Hamiltonian instances, with error bars indicating one
standard deviation across these instances.

\begin{figure}[t]
\centering
\includegraphics[width=13.5cm]{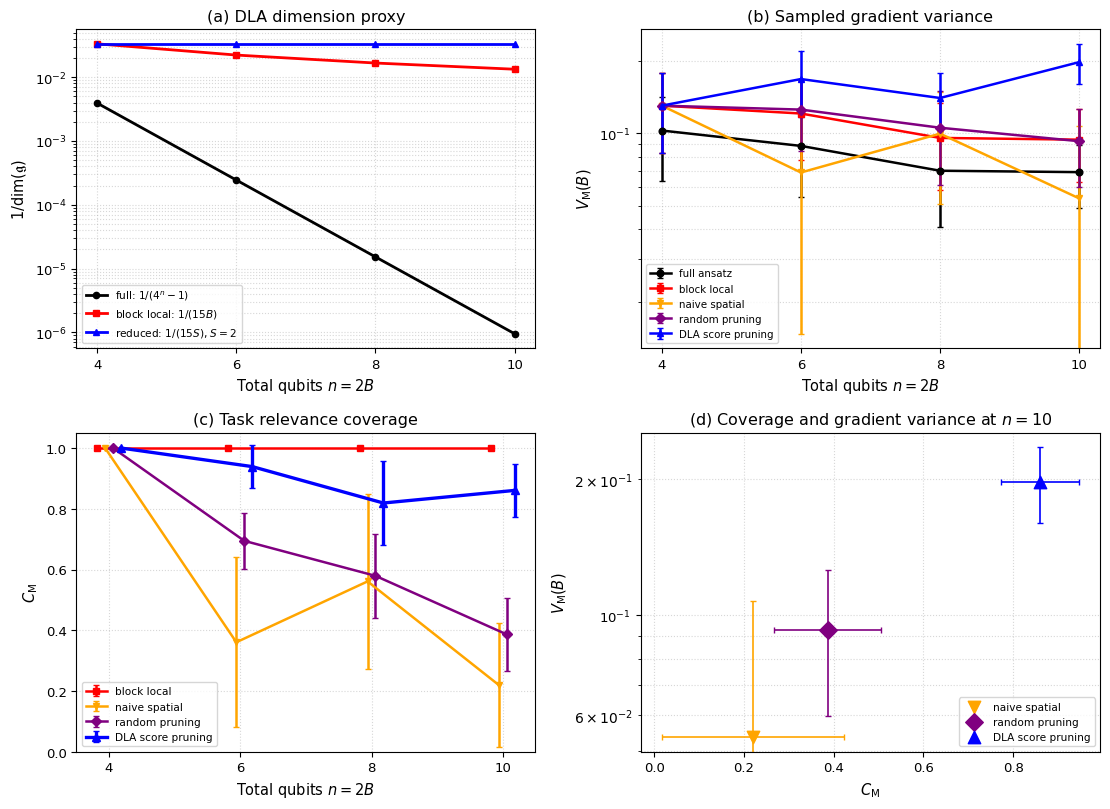}
\caption{
DLA reduction for trainability.
DLA dimension factor \(1/\dim(\mathfrak g)\) for the full, block local,
and \(S\)-block ansatzes.
(b) Sampled gradient variance \(V_{\mathsf M}(B)\) for the loss in
Eq.~\eqref{eq:random-hamiltonian-loss}.
(c) Task relevance coverage \(C_{\mathsf M}\) defined in
Eq.~\eqref{eq:coverage}.
(d) Coverage and gradient variance for the pruning methods at the
largest simulated system size, \(n=10\).
Error bars indicate one standard deviation over Hamiltonian instances.
For random pruning, each instance value is first averaged over five
independently sampled block sets.
}
\label{fig:dla-bp-pruning}
\end{figure}

Fig.~\ref{fig:dla-bp-pruning} summarizes the benchmark. Panel~(a) isolates the algebraic dimension factor \(1/\dim(\mathfrak g)\).
For the full ansatz, this quantity scales as \(1/(4^n-1)\) and therefore
decreases exponentially with the system size. Restricting the dynamics to the
block local algebra changes this dependence to \(1/(15B)\), while retaining
only \(S\) blocks gives \(1/(15S)\).
Since \(S=2\) is fixed in the present benchmark, the reduced DLA dimension
remains constant as the total number of blocks increases. 
Panel~(b) shows the sampled gradient variance \(V_{\mathsf M}(B)\). The black
full ansatz curve decreases to a comparatively small variance scale as the
system size increases, whereas the blue DLA score pruning curve remains at a
larger scale. The orange naive spatial pruning curve exhibits substantial
fluctuations, while the purple random pruning and red block local curves lie
between these cases.
Thus, among the pruning methods considered, DLA score pruning maintains a
comparatively large gradient scale. The finite size results are used as a
numerical trainability diagnostic.

Panel~(c) reports the task relevance coverage \(C_{\mathsf M}\). The red block
local curve remains at \(C_{\mathsf M}=1\), since all block components are
retained. Among the methods restricted to \(S\) blocks, the blue DLA score
pruning curve remains highest, whereas the purple random pruning and orange
naive spatial pruning curves retain smaller fractions and exhibit larger
instance to instance variation.
This follows from the DLA score rule, which
selects the \(S\) components with the largest \(R_j\) and therefore retains the
largest Hamiltonian projection weight.
Panel~(d) combines the two diagnostics at \(n=10\). The blue DLA score pruning
point has both larger coverage and larger sampled gradient variance than the
purple random pruning and orange naive spatial pruning points. This indicates a
favorable balance between reducing the DLA dimension, retaining task relevant
Hamiltonian components, and maintaining the gradient scale in the present
benchmark. The comparison is specific to the Hamiltonian ensemble and circuit
construction considered here.

Overall, Fig.~\ref{fig:dla-bp-pruning} shows that the comparison with random
and naive pruning highlights that dimension reduction alone is insufficient;
the choice of retained DLA components is crucial for preserving both task
relevance and gradient scale. This observation is consistent with the
Lie algebraic perspective that trainability depends not only on the DLA
dimension, but also on how the task Hamiltonian is represented within the
relevant algebra~\cite{ragone2024lie}.

\section{Conclusion and Discussion}

In this work, we have developed a finite-dimensional Lie-algebraic framework
for modifying dynamical generating sets of Hamiltonian-driven quantum systems.
The framework is organized around three operations on DLAs: composition,
invariance, and reduction. For DLA composition, we construct a block-projector
realization of direct sums of component DLAs represented on a common system
Hilbert space. For DLA invariance, we study how changes in a generating set
affect the resulting Lie closure. This includes a nearest-neighbor
Pauli-string construction for \(\mathfrak{su}(2^N)\) and algebraic
diagnostics for the case where additional generators are added. For DLA
reduction, we consider compact reductive DLAs with decompositions into simple
ideals. Projection onto selected simple ideals gives reduced generating sets
whose Lie closures are the corresponding ideal sums.

The examples and numerical studies illustrate how these algebraic operations
can be used in finite-dimensional quantum-control and variational settings.
The direct-sum construction gives a register-level organization of different
dynamical branches. The central-spin example shows how projection overlap and
DLA percentage change capture different aspects of generator modification. The
block-local Hamiltonian example shows how DLA-based selection can be used to
compare reduced ansatz structures. These results also indicate possible connections with quantum engineering
tasks. DLA composition is related to parallel representations of dynamical
branches. DLA invariance provides algebraic criteria for modifying generator
sets under hardware constraints. DLA reduction can guide circuit or ansatz
pruning when the task-relevant dynamics lies in a selected ideal component.

\begin{figure}[htbp]
\centering
\includegraphics[width=9cm]{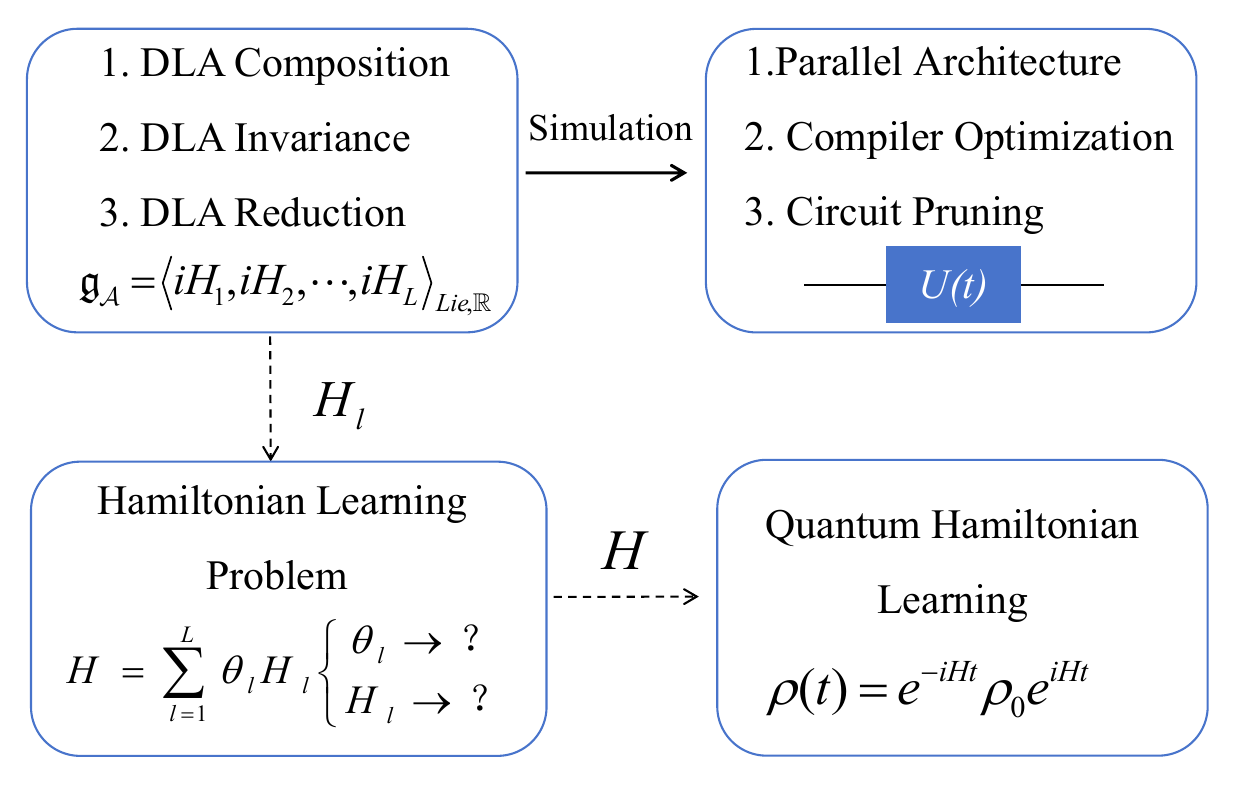}
\caption{Our framework and possible downstream tasks in quantum control, circuit design,
and Hamiltonian learning.}
\label{discussion}
\end{figure}

The schematic perspective in Fig.~\ref{discussion} also points to two
extensions of the present framework. On the simulation side, recent work has
shown that Hamiltonian simulation can benefit from additional structure, such
as compressed shadow-state representations, random-input or average-case
settings, and entanglement properties of the simulated state
\cite{somma2025shadow,zhao2022hamiltonian,zhao2025entanglement}. The DLA
viewpoint developed here may provide complementary algebraic information about
which Hamiltonian directions are available, which generator modifications
preserve the Lie closure, and which ideal components support the relevant
dynamics. Combining this information with implementation-oriented techniques,
such as resource-efficient linear-combination-of-unitaries methods
\cite{chakraborty2024implementing}, is a possible route toward connecting
generator-set design with Hamiltonian synthesis cost and control overhead. On
the learning side, experimental Hamiltonians often take the form
\(H=\sum_{l=1}^{L}\theta_l H_l\), where the interaction terms \(H_l\), the
coefficients \(\theta_l\), or both may be partially unknown. Hamiltonian
learning protocols aim to infer such dynamical models from experimental data
or state trajectories
\cite{wang2017experimental,gu2024practical,heightman2025solving}. In this
setting, DLA structure may serve as an algebraic prior for constraining the
learning problem and for checking whether a learned Hamiltonian model is
compatible with the available control architecture.

\section*{Acknowledgments}
We thank Haoran Zhu for helpful discussion.
This work is supported by the National Natural Science Foundation of China (62471187 and 12371458), Guangdong Basic and Applied Basic Research
Foundation (2026A1515011707).
This work is also supported by China Scholarship Council.

\section*{Data availability statement}
The data are available from the corresponding author upon
reasonable request.

\bibliography{Ref}

\appendix

\section{Proof of Theorem~\ref{composition1} in DLA composition}
\label{app:proof-composition}

\begin{proof}
Denote
\( \mathcal{A}_{\pi_m}
    =
    \left\{
        A_{m,l}\otimes\Pi_m
        \mid l=1,2,\ldots,L_m
    \right\}.\)
We first establish the commutation of different projector sectors.
By definition, orthogonal projectors satisfy
\(\Pi_i\Pi_j=0\) for \(i\neq j\). For any
\(A\in\mathcal A_i\) and \(B\in\mathcal A_j\), where \(i\neq j\), we have
\([A\otimes\Pi_i,B\otimes\Pi_j]
    =
    0,\)
ensuring elements from different \(\mathcal{A}_{\pi_m}\) commute trivially,
and no additional Lie-algebra directions are generated by commutators
between different sectors.

Based on \(\mathcal A_{\pi_m}\), we show that the generated DLA
\[
    \mathfrak g_{\mathcal A_{\pi_m}}
    =
    \langle\mathcal A_{\pi_m}\rangle_{\mathrm{Lie},\mathbb R}
\]
is isomorphic to \(\mathfrak g_{\mathcal A_m}\) for \(1\leq m\leq K\).
Define a linear map
\[
    \phi_m:\mathfrak g_{\mathcal A_m}
    \to
    \mathfrak g_{\mathcal A_{\pi_m}},
    \quad
    \phi_m(X)=X\otimes\Pi_m .
\]
For any \(X,Y\in\mathfrak g_{\mathcal A_m}\), the commutator transforms as
\begin{equation}
    \begin{aligned}
    \phi_m([X,Y])
    &=
    [X,Y]\otimes\Pi_m  \\
    &=
    [X\otimes\Pi_m,Y\otimes\Pi_m]  \\
    &=
    [\phi_m(X),\phi_m(Y)] .
    \end{aligned}
\end{equation}
Here we used the projection property \(\Pi_m^2=\Pi_m\).
Since $\phi_m$ is a linear bijection onto its image generated by $\mathcal{A}_{\pi_m}$, it is a Lie algebra isomorphism, implying 
\(\mathfrak g_{\mathcal A_{\pi_m}}
    \cong
    \mathfrak g_{\mathcal A_m}.\)

Consequently, we can explicitly construct the basis. Let
\(\{U_{m,1},\ldots,U_{m,\delta_m}\}\) be a basis of
\(\mathfrak g_{\mathcal A_m}\). By the isomorphism \(\phi_m\), the set
\begin{equation}
    \mathcal B_m
    =
    \{
        U_{m,p}\otimes\Pi_m
        \mid p=1,\ldots,\delta_m
    \}
\end{equation}
constitutes a basis for \(\mathfrak g_{\mathcal A_{\pi_m}}\).

Finally, we confirm the direct-sum structure. The DLA
\(\mathfrak g_{\mathcal A'}\) is spanned by
\(\bigcup_{m=1}^K\mathcal B_m\), since
\[
    \mathcal A'
    =
    \bigcup_{m=1}^K \mathcal A_{\pi_m},
\]
and each \(\mathcal B_m\) spans the DLA
\(\mathfrak g_{\mathcal A_{\pi_m}}\). By orthogonality of
\(\{\Pi_m\}\), commutators between \(\mathcal B_m\) and
\(\mathcal B_n\) for \(m\neq n\) vanish with
\[
    [U_{m,p}\otimes\Pi_m,U_{n,q}\otimes\Pi_n]
    =
    U_{m,p}U_{n,q}\otimes\Pi_m\Pi_n
    -
    U_{n,q}U_{m,p}\otimes\Pi_n\Pi_m
    =
    0 .
\]
This ensures that the subalgebras
\(\mathfrak g_{\mathcal A_{\pi_m}}\) are mutually commutative and their subspaces are disjoint. 

Since
\([U_s \otimes \Pi_i, U_t \otimes \Pi_j] = [U_s,U_t] \otimes \Pi_i,\)
for all $s,t\in \{1,2,\cdots,\delta\}$ when $i=j$, otherwise \( [U_s \otimes \Pi_i, U_t \otimes \Pi_j] =0\). So in general, if \( \{U_1, \ldots, U_\delta\} \) is a basis for a DLA \( \mathfrak{g} \), we have \( \bigcup_{j=1}^K \{U_1 \otimes \Pi_j, \ldots, U_\delta \otimes \Pi_j\} \) constitutes a basis for a DLA isomorphic to \( \bigoplus_{j=1}^K \mathfrak{g} \). Extending this to multiple DLAs, we have  \(\bigcup_{m=1}^K \mathcal{B}_m\) is a basis for \(\mathfrak{g}_{\mathcal{A}'}\). The mutual commutativity of \(\mathfrak{g}_{\mathcal{A}_{\pi_m}}\) ensures that \(\mathfrak{g}_{\mathcal{A}'}\) is the internal direct sum of these subalgebras.
Therefore,
\[
    \mathfrak g_{\mathcal A'}
    =
    \bigoplus_{m=1}^K
    \mathfrak g_{\mathcal A_{\pi_m}}
    \cong
    \bigoplus_{m=1}^K
    \mathfrak g_{\mathcal A_m}.
\]
\end{proof}

\section{Proofs for DLA invariance}
\subsection{Proof of Proposition~\ref{prop:su2-case}}\label{proof-of-invariance-prop}

\begin{proof}
Since
\[
i\mathcal{P}_{2}^{*}
=
\{i A \otimes B \mid A,B \in\{I,X,Y,Z\}\}\backslash\{iI\otimes I\},
\]
it is enough to show that all non-identity two-qubit Pauli strings are
contained in \(\langle\mathcal{A}_2\rangle_{\mathrm{Lie},\mathbb R}\).
The local generators \(i\sigma^x_1,i\sigma^y_1,i\sigma^x_2,i\sigma^y_2\)
already generate the local Pauli strings
\(i\sigma^\alpha_1\) and \(i\sigma^\beta_2\), with
\(\alpha,\beta\in\{x,y,z\}\). Moreover,
\[
[i\sigma^x_1,i\sigma^y_1\sigma^y_2]
=
-2i\sigma^z_1\sigma^y_2,
\]
\[
[i\sigma^y_1,[i\sigma^x_1,i\sigma^y_1\sigma^y_2]]
=
4i\sigma^x_1\sigma^y_2 .
\]
Together with \(i\sigma^y_1\sigma^y_2\), this gives
\(i\sigma^\alpha_1\sigma^y_2\) for
\(\alpha\in\{x,y,z\}\). Commuting these elements with the local generators
on the second qubit gives
\(i\sigma^\alpha_1\sigma^\beta_2\) for all
\(\alpha,\beta\in\{x,y,z\}\). Hence
\(\langle\mathcal{A}_2\rangle_{\mathrm{Lie},\mathbb R}\) contains all
Pauli strings in \(i\mathcal{P}_{2}^{*}\), and therefore
\[
\langle\mathcal{A}_2\rangle_{\mathrm{Lie},\mathbb R}
=
\mathfrak{su}(4).
\]
\end{proof}

\subsection{Proof of Theorem~\ref{su2N-unchanged}}
\label{proof-of-invariance-thm}

\begin{proof}
Recall that \(\mathcal B_I'\) denotes the non-local Pauli-string
generator set defined in Eq.~\eqref{b1}, while
\(\mathcal B_{II}'\) denotes the nearest-neighbour generator set defined in Eq.~\eqref{b2}.

When $i=4$, we can compute the element of $\mathcal{B}_I'$ in Eq.(\ref{b1}) satisfying
\[i\sigma^x_2 \otimes \sigma^z_3 \otimes \sigma^y_4 \propto [i\sigma^x_2\otimes\sigma^y_3,i\sigma^x_3\otimes\sigma^y_4];\]
 \[i\sigma^y_2 \otimes \sigma^z_3 \otimes \sigma^x_4 \propto [i\sigma^y_2\otimes\sigma^x_3,i\sigma^y_3\otimes\sigma^x_4].\]
By induction method, using the nested Lie bracket in the definition of DLA, we further have
\[i\sigma^x_2 \otimes \prod_{2<j<i}\sigma^z_j \otimes \sigma^y_i \propto [i\sigma^x_2\otimes\sigma^y_3,\cdots,[i\sigma^x_{i-2}\otimes\sigma^y_{i-1},i\sigma^x_{i-1}\otimes\sigma^y_i]];\]
\[i\sigma^y_2 \otimes \prod_{2<j<i}\sigma^z_j \otimes \sigma^x_i \propto [i\sigma^y_2\otimes\sigma^x_3,\cdots,[i\sigma^y_{i-2}\otimes\sigma^x_{i-1},i\sigma^y_{i-1}\otimes\sigma^x_i]].\]
Then we have \(\mathcal{B}_I' \subset \langle \mathcal{B}_{II}' \rangle_{Lie,\mathbb{R}}\) in Eq.(\ref{b1}) and Eq.(\ref{b2}).
Combined with Proposition \ref{prop:su2-case} and Lemma \ref{lemma3}, the modified set \(\mathcal{A}'\) generates the full $\mathfrak{su}(2^N)$.
\end{proof}

\section{Proof in DLA reduction}\label{proof-of-reduction}

\subsection{Proof of Lemma \ref{new-lemma-for-reduction}}

\begin{proof}
It suffices to show that $\langle\im(\ad_F)\rangle_{\Lie,\mathbb R}$ is a non-zero ideal of $\mathfrak g$.
Let $C\in\ker(\ad_F)$ and $X\in\mathfrak g$.
By the Jacobi identity,
\[
    [C,[F,X]]=[[C,F],X]+[F,[C,X]]=[F,[C,X]]\in\im(\ad_F).
\]
Thus, $[\ker(\ad_F),\im(\ad_F)]\subseteq\im(\ad_F)$.
Since $\ad_C$ is a derivation, it follows that
\[
    [\ker(\ad_F),\langle\im(\ad_F)\rangle_{\Lie,\mathbb R}]
    \subseteq
    \langle\im(\ad_F)\rangle_{\Lie,\mathbb R}.
\]
Also,
\[
    [\im(\ad_F),\langle\im(\ad_F)\rangle_{\Lie,\mathbb R}]
    \subseteq
    \langle\im(\ad_F)\rangle_{\Lie,\mathbb R}.
\]
Using $\mathfrak g=\ker(\ad_F)\oplus\im(\ad_F)$, we obtain
\[
    [\mathfrak g,\langle\im(\ad_F)\rangle_{\Lie,\mathbb R}]
    \subseteq
    \langle\im(\ad_F)\rangle_{\Lie,\mathbb R}.
\]
Hence $\langle\im(\ad_F)\rangle_{\Lie,\mathbb R}$ is an ideal of $\mathfrak g$.
It is non-zero, since $\im(\ad_F)\ne0$.
By simplicity, it is equal to $\mathfrak g$.
Since $\im(\ad_F)=[F,\mathfrak g]$, the result follows.
\end{proof}

\subsection{Proof of Proposition \ref{prop2}}

\begin{proof}
Let $\langle\cdot,\cdot\rangle$ be a positive definite invariant inner product on $\mathfrak g$.
For $X,Y\in\mathfrak g$, we have $\langle [F,X],Y\rangle=-\langle X,[F,Y]\rangle$, so $(\ad_F)^*=-\ad_F$.
Hence $\im(\ad_F)=(\ker(\ad_F)^*)^{\perp}=\ker(\ad_F)^{\perp}$.
Thus $\mathfrak g=\ker(\ad_F)\oplus\im(\ad_F)$.
Since $Z(\mathfrak g)=0$, the assumption $F\ne0$ gives $\im(\ad_F)\ne0$.
The result follows from Lemma~\ref{new-lemma-for-reduction}.
\end{proof}

\subsection{Proof of Lemma \ref{lemma5}}

\begin{proof}
We argue by induction on $m$.
The case $m=1$ is clear.
Put $\mathfrak a=\bigoplus_{i=1}^{m-1}\mathfrak s_i$.
By induction, the projection of $\mathfrak q$ onto $\mathfrak a$ is $\mathfrak a$.
Let
\[
    K=\{X\in\mathfrak s_m:(0,X)\in\mathfrak q\}.
\]
Then $K$ is an ideal of $\mathfrak s_m$.
Hence $K=0$ or $K=\mathfrak s_m$.
If $K=\mathfrak s_m$, then $\mathfrak q=\mathfrak a\oplus\mathfrak s_m$.
Suppose $K=0$.
Then $\mathfrak q$ is the graph of a surjective Lie algebra homomorphism $\varphi:\mathfrak a\to\mathfrak s_m$.
The images of distinct simple summands of $\mathfrak a$ commute, so exactly one summand maps non-trivially.
Thus $\varphi|_{\mathfrak s_i}:\mathfrak s_i\to\mathfrak s_m$ is an isomorphism for some $i<m$.
The projection of $\mathfrak q$ onto $\mathfrak s_i\oplus\mathfrak s_m$ is then the graph of this isomorphism, contrary to the hypothesis.
Therefore $K=\mathfrak s_m$, and the result follows.
\end{proof}

\subsection{Proof of Proposition \ref{prop4}}

\begin{proof}
By Proposition~\ref{prop3}, $\mathfrak g_{\mathcal A'_F}\subseteq\mathfrak s$.
For $j\in S$,
\[
    \pi_j([F,A_l])=[F_j,\pi_j(A_l)].
\]
Hence the first assumption says that every coordinate projection of $\mathfrak g_{\mathcal A'_F}$ onto $\mathfrak g_j$, $j\in S$, is surjective.
Similarly, for distinct $j,k\in S$, the projection of $\mathfrak g_{\mathcal A'_F}$ onto $\mathfrak g_j\oplus\mathfrak g_k$ is generated by
\[
    \big([F_j,\pi_j(A_l)],[F_k,\pi_k(A_l)]\big),
    \quad
    1\le l\le L.
\]
Thus the second assumption gives surjectivity on every pair of isomorphic simple factors.
Lemma~\ref{lemma5} now gives $\mathfrak g_{\mathcal A'_F}=\mathfrak s$.
\end{proof}

\subsection{Proof of Proposition \ref{prop5}}

\begin{proof}
Since $\rho(A_l)\in\mathfrak s$ for all $l$,
\[
    \big\langle \rho(A_l):1\le l\le L\big\rangle_{\Lie,\mathbb R}
    \subseteq\mathfrak s.
\]
As $\rho$ is a Lie algebra homomorphism and $\mathfrak g_{\mathcal A}=\langle A_1,\ldots,A_L\rangle_{\Lie,\mathbb R}$, we have
\[
    \rho(\mathfrak g_{\mathcal A})
    =
    \big\langle \rho(A_l):1\le l\le L\big\rangle_{\Lie,\mathbb R}.
\]
The restriction of $\rho$ to $\mathfrak s$ is the identity, so $\mathfrak s\subseteq\rho(\mathfrak g_{\mathcal A})$.
Since $\rho(\mathfrak g_{\mathcal A})\subseteq\mathfrak s$, equality follows.
\end{proof}

\subsection{Proof of Theorem \ref{thm-subalgebra-modify}}

\begin{proof}
Both $\mathfrak s$ and $\mathfrak k$ are ideals of $\mathfrak g_{\mathcal A}$.
Hence $[\mathfrak s,\mathfrak k]\subseteq\mathfrak s\cap\mathfrak k=0$.
For $X=X_\mathfrak s+X_\mathfrak k$ and $Y=Y_\mathfrak s+Y_\mathfrak k$, with $X_\mathfrak s,Y_\mathfrak s\in\mathfrak s$ and $X_\mathfrak k,Y_\mathfrak k\in\mathfrak k$, we have
\[
    \pi_\mathfrak s([X,Y])
    =
    \pi_\mathfrak s([X_\mathfrak s,Y_\mathfrak s]+[X_\mathfrak k,Y_\mathfrak k])
    =
    [X_\mathfrak s,Y_\mathfrak s]
    =
    [\pi_\mathfrak s(X),\pi_\mathfrak s(Y)].
\]
Thus $\pi_\mathfrak s$ is a Lie algebra homomorphism.
It is the identity on $\mathfrak s$, and the claim follows from Proposition~\ref{prop5}.
\end{proof}

\section{Further implementation details}
\label{app:numerical-procedures}

\subsection{Numerical algebraic procedures}

Given a finite generating set
\(\mathcal G=\{G_1,\ldots,G_m\}\), we first construct an orthonormal basis
\(\mathcal B\) of \(\operatorname{span}_{\mathbb R}\mathcal G\).
For each pair \(B_p,B_q\in\mathcal B\), the commutator
\([B_p,B_q]\) is orthogonalized against the current basis.
If the norm of the residual exceeds the numerical threshold, the normalized
residual is appended to \(\mathcal B\), and commutators involving the new basis
element are subsequently evaluated. The procedure terminates when all such
commutators lie in the current span within tolerance. The Lie-algebra dimension
is then \(|\mathcal B|\).
Numerical linear independence and matrix ranks are determined using
orthogonalization and singular-value thresholding.

For a matrix set
\(\mathcal S=\{S_1,\ldots,S_r\}\subset\mathbb C^{d\times d}\), the commutant
basis is obtained from the null space of the stacked matrix whose \(j\)th block
is
\[
S_j^T\otimes I-I\otimes S_j.
\]
If \(\{C_\alpha\}\) is a basis of the commutant, we write
\(Z=\sum_\alpha c_\alpha C_\alpha\) and solve the stacked homogeneous linear
system obtained from \([Z,C_\beta]=0\) for all \(\beta\). The resulting null
space gives the center.
Matrix ranks are determined by singular-value decomposition, with singular
values larger than
\(\max\{\tau_{\mathrm{abs}},\tau_{\mathrm{rel}}\sigma_{\max}\}\)
counted as nonzero.

\subsection{Finite-dimensional checks of the algebraic constructions}
For the composition example, the computed dimensions are
\(\dim\mathfrak g_{\mathcal A}=4\),
\(\dim\mathfrak g_{\mathcal B}=1\), and
\(\dim\mathfrak g_{\mathcal A'}=5\), consistent with
\(\mathfrak g_{\mathcal A'}
\cong
\mathfrak g_{\mathcal A}\oplus\mathfrak g_{\mathcal B}\).

For the central-spin example, the calculation gives
\(\dim\langle\mathcal P\rangle_{\mathrm{Lie},\mathbb R}=4\) and
\(\dim\!\left(Z(\mathcal P')\cap\mathfrak{su}(8)\right)=2\).
The fixed-reference projection matrix has
\(\operatorname{rank}(\tilde T)=1\). For
\(\epsilon=0,0.3,0.6\), the corresponding ranks of \(T_\epsilon\) are
\(1,2,2\), while the Lie-closure dimensions of
\(\langle\mathcal P\cup\mathcal Q_\epsilon\rangle_{\mathrm{Lie},\mathbb R}\)
are \(4,5,5\).

For the DLA reduction construction, a representative instance with
\(B=3\) and \(S=2\) gives
\(\dim\langle\mathcal A\rangle_{\mathrm{Lie},\mathbb R}=45=15B\) and
\(\dim\langle\mathcal A_{\mathrm{red}}\rangle_{\mathrm{Lie},\mathbb R}
=30=15S\).

The two finite-generator filtering examples were evaluated using the same
Lie-closure procedure. In the first example, the filtered dimensions are
\(1\) for \(F=iX_1\) and \(3\) for \(F=iZ_1\). In the second, the symmetric
choice \(F=iZ_1+iZ_2\) gives dimension \(3\), whereas
\(F=iZ_1+2iZ_2\) gives dimension \(6\).

\nocite{*}

\end{document}